\documentclass{article}

\usepackage[english]{babel}

\usepackage[letterpaper,top=2cm,bottom=2cm,left=3cm,right=3cm,marginparwidth=1.75cm]{geometry}

\usepackage{amsthm,amssymb,amsmath}  
\usepackage{mathtools}
\usepackage{graphicx}
\usepackage[colorlinks=true, allcolors=blue]{hyperref}
\usepackage[capitalize]{cleveref}
\usepackage{xspace}
\usepackage{thmtools}
\usepackage{thm-restate}
\usepackage{url}
\usepackage{authblk}
\usepackage{comment}
\usepackage{subcaption}
\usepackage{bm}
\usepackage{multirow}

\def\dd{\mathinner{.\,.}}
\newcommand{\Depth}{\textsf{\small Depth}\xspace}

\newcommand{\Witness}{\textsf{\small Witness}\xspace}

\newcommand{\Lcode}{\textsf{\small LastCode}\xspace}
\newcommand{\LcodeInt}{\textsf{\small LastCodeInt}\xspace}
\newcommand{\Pcode}{\textsf{\small PrefCode}\xspace}
\newcommand{\Scode}{\textsf{\small SufCode}\xspace}
\newcommand{\OPST}{\textsf{\small OPST}\xspace}

\newcommand{\OPSTMP}{\textsf{\small OPST-MP}\xspace}
\newcommand{\OPSTCP}{\textsf{\small OPST-CP}\xspace}

\newcommand{\BAMP}{\textsf{\small BA-MP}\xspace}
\newcommand{\BACP}{\textsf{\small BA-CP}\xspace}

\newcommand{\MOPP}{\textsf{\small MOPP}\xspace}
\newcommand{\MOPPM}{\textsf{\small MOPP-Miner}\xspace}

\newcommand{\HOU}{\textsc{HOU}\xspace}
\newcommand{\SOL}{\textsc{SOL}\xspace}
\newcommand{\ECG}{\textsc{ECG}\xspace}
\newcommand{\TRA}{\textsc{TRA}\xspace}
\newcommand{\TEM}{\textsc{TEM}\xspace}

\newcommand{\WHA}{\textsc{WHA}\xspace}

\newcommand{\CCT}{\textsc{CCT}\xspace}
\newcommand{\WAF}{\textsc{WAF}\xspace}

\newcommand{\dollar}{\texttt{\$}\xspace}
\newcommand{\Locus}{\textsf{\small Locus}\xspace}
\newcommand{\Fdown}{\textsf{\small FirstDown}\xspace}
\newcommand{\Slink}{\textsf{\small SufLink}\xspace}

\newcommand{\cO}{\mathcal{O}}
\usepackage{algorithm}
\usepackage[noend]{algpseudocode}

\newtheorem{theorem}{Theorem}

\newtheorem{lemma}{Lemma}

\newtheorem{example}{Example}
\newtheorem{observation}{Observation}

\title{Scalable Order-Preserving Pattern Mining}
\author[1]{Ling Li$^*$}
\author[2]{Wiktor Zuba$^*$}
\author[1]{Grigorios Loukides}
\author[2,3]{Solon P.\ Pissis}
\author[4]{\\Maria Matsangidou}
\affil[1]{King's College London, London, UK}
\affil[2]{CWI, Amsterdam, The Netherlands}
\affil[3]{Vrije Universiteit, Amsterdam, The Netherlands}
\affil[4]{CYENS--Centre of Excellence, Nicosia, Cyprus}

\begin{document}

\maketitle
\renewcommand\thefootnote{} 
\footnotetext{$^*$ denotes equal contribution.}

\begin{abstract}
    Time series are ubiquitous in domains ranging from medicine to marketing and finance. Frequent Pattern Mining (FPM) from a time series has thus received much attention. This general problem has been studied under different matching relations determining whether two time series match or not. Recently, it has been studied under the \emph{order-preserving} (OP) matching relation stating that a match occurs when two time series have the same relative order (i.e., ranks) on their elements. Thus, a \emph{frequent OP pattern} captures a trend shared by sufficiently many parts of the input time series. Here, we propose exact, highly scalable algorithms for FPM in the OP setting. Our algorithms employ an \emph{OP suffix tree} (OPST) as an index to store and query time series efficiently. Unfortunately, there are no practical algorithms for OPST construction. Thus, we first propose a novel and practical $\cO(n\sigma\log \sigma)$-time and $\cO(n)$-space algorithm for constructing the OPST of a length-$n$ time series over an alphabet of size $\sigma$. 
    We also propose an alternative faster OPST construction algorithm running in $\cO(n\log \sigma)$ time using $\cO(n)$ space; this algorithm is mainly of theoretical interest. 
    Then, we propose an exact $\cO(n)$-time and $\cO(n)$-space algorithm for mining all \emph{maximal} frequent OP patterns, given an OPST. This significantly improves on the state of the art, which takes $\Omega(n^3)$ time in the worst case. We also formalize the notion of \emph{closed} frequent OP patterns and propose an exact $\cO(n)$-time and $\cO(n)$-space algorithm for mining all closed frequent OP patterns, given an OPST. We conducted experiments using real-world, multi-million letter time series showing that our $\cO(n\sigma \log \sigma)$-time OPST construction algorithm runs in $\cO(n)$ time on these datasets despite the $\cO(n\sigma \log \sigma)$ bound; that  our frequent pattern mining algorithms are up to orders of magnitude faster than the state of the art and natural Apriori-like baselines; and that OP pattern-based clustering is effective. 
\end{abstract}

\section{Introduction}\label{sec:intro}

A \emph{time series} is a sequence of data points indexed by time,  
which are often recorded at successive, equally-spaced points in time. As a data type, a time series is thus simply a \emph{string} over the alphabet of real numbers~\cite{acmsurv1}, integers~\cite{tm_book}, or even characters (e.g., of ``strong'', ``medium'', or ``weak'' type~\cite{ntpminer_tkdd}). 
Many application domains feature time series, including medicine, where they model e.g., electrocardiogram (ECG) recordings or electroencephalography (EEG) data~\cite{DBLP:journals/pvldb/0001L22};  
sensor networks, where they model e.g., sensor readings~\cite{sdmts}; and finance, where they model e.g., quarterly revenue, monthly sales, or stock prices~\cite{tkde_opp}.

In all these application domains, mining a time series is useful, as it can discover actionable patterns,  such as 
heart beats in ECG data~\cite{DBLP:journals/pvldb/0001L22}, sleep spindles in EEG sleep data~\cite{DBLP:journals/pvldb/0001L22};  electricity consumption profiles of devices, or correlations between traffic jams in sensor data~\cite{vldb21}; and interesting stock  co-movements in financial data~\cite{tan_tkde}. Towards this goal, the problem of pattern mining (a.k.a motif discovery) from a time series has been studied for over twenty years (see~\cite{acmsurv1} for a survey). This general problem asks for \emph{patterns} (e.g., substrings, subsequences, or trends), which occur frequently in an input time series, based on a \emph{matching relation} which dictates 
whether two time series match or not.  

There are numerous algorithms for dealing with this problem  
most of which fall into two categories based on the matching relation they adopt. The algorithms in the first category (e.g.,~\cite{ntpminer_tkdd,bigdata23,eswa1})  are either applied to symbolic sequences~\cite{ntpminer_tkdd} (i.e., strings of characters) directly, or first 
transform the (numerical) input time series into symbolic sequences~\cite{bigdata23,eswa1}. They all   
adopt \emph{string matching} relations which state that all occurrences of a pattern need to have exactly the same characters in the same order to match, 
except possibly wildcard characters. 
To discover the patterns, these algorithms resort to frequent string~\cite{takeaki} or frequent 
sequence~\cite{DBLP:journals/csur/MabroukehE10} mining. 
Thus, they  typically find substrings or subsequences  occurring sufficiently frequently. The algorithms in the second category (e.g., \cite{valmod,DBLP:journals/pvldb/0001L22,distance_dtw1, distance_dtw3_kdd}) adopt a \emph{distance-based matching} relation stating that all occurrences of a pattern need to be ``similar'', according to a distance function, to match. These algorithms typically find the largest set of patterns that match according to the matching relation. The drawback of both categories of algorithms is that they may not capture \emph{trends} effectively. This is because such trends are often not preserved well after the transformation of the input time series into a symbolic sequence~\cite{tkde_opp,xindong_opp}; or because such trends may not be captured by patterns which are at small distance but nonetheless correspond to very different trends. 

In response, a very recent class of pattern mining methods~\cite{xindong_opp, tkde_opp, tkdd_opp, topkopp}  utilizes the well-established \emph{order-preserving} (OP) matching relation~\cite{DBLP:journals/ipl/KubicaKRRW13,
DBLP:journals/tcs/CrochemoreIKKLP16}. 
This property states that two time series are \emph{order-preserving}  if and only if they have the same relative order (i.e., ranks) on their elements. 
For example, the  time series  
$X=(4, 2, 5, 5, 1)$ and $Y=(5, 2, 7, 7, 0)$
are OP, as they have the same sequence of ranks 
$(3,2,4,4,1)$  
on their elements 
(e.g., $4$ and $5$ is \emph{the third smallest} element in $X$ and $Y$, respectively).  
This sequence of ranks is called an \emph{OP pattern}. Thus, an OP pattern represents the collection of \emph{all} time series with the same trend, captured by the sequence of their ranks. This is useful to derive the ``shape'' of these time series, even when there are small fluctuations in the data, as is common in practice~\cite{acmsurv1}. The classic \emph{frequent pattern mining} (FPM) problem extends naturally to the OP setting~\cite{xindong_opp,tkde_opp}: Given a string $w$ of length $n$ over a totally ordered alphabet of size $\sigma$ and a frequency (minimum support) threshold $\tau$, return all OP patterns with frequency at least $\tau$ in $w$. 

\begin{example}
\label{example1} 
The time series $S$ below models the average monthly price of a stock from the electric power sector over a 20-month period. The time series 
 $S_1=(56, 57, 62, 59, 58)$ and $S_2=(63, 64, 68, 67, 66)$  are OP because they have the same sequence of ranks on their elements, i.e., OP pattern, $p=(1, 2, 5, 4 , 3)$ and hence the same trend (i.e., the price increased as winter approached and then decreased as spring approached). Pattern $p$ is frequent for $\tau=2$, as it corresponds to two time series, $S_1$ and $S_2$, which are part of $S$. 

\begin{figure}[!ht]
\centering 
\includegraphics[width=0.65\linewidth]{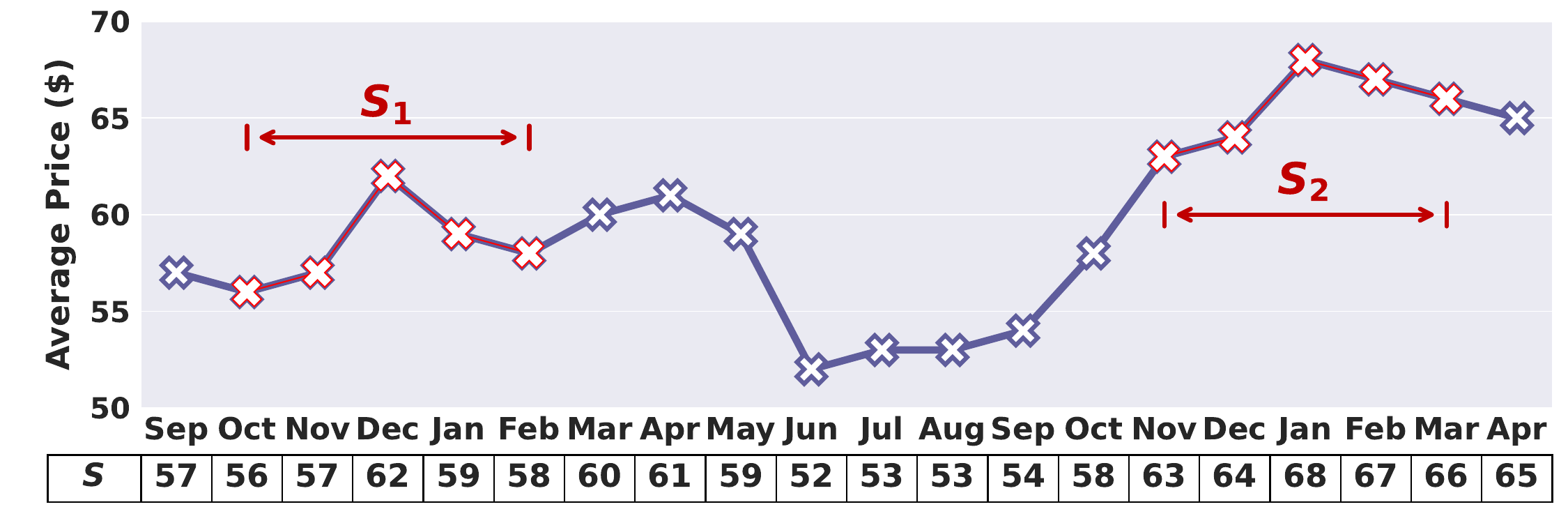}
\end{figure}

\end{example}

A prime example among the class of OP pattern mining methods is the  \MOPPM~\cite{xindong_opp} algorithm (see Section~\ref{sec:related}), which mines  \emph{maximal} frequent OP patterns (i.e., those that cannot be extended to either side and still remain frequent). 

\begin{example}[cont'd from Example~\ref{example1}] 
The OP pattern $p$ in Example~\ref{example1} is maximal because it cannot be extended to either side and still have frequency at least $\tau=2$. Specifically,
if we prepend to $S_1$ the element (letter) on its left and do the same for $S_2$ we obtain $S'_1=(57,56, 57, 62, 59, 58)$ and $S'_2=(58, 63, 64, 68, 67, 66)$ whose OP patterns are 
$p'_1=(2,1,2,5,4,3)$ and $p'_2=(1,2,3,6,5,4)$, respectively, and have frequency $1 < \tau$. Similarly, if we 
append to $S_1$ the letter on its right and do the same for $S_2$, we obtain $S''_1=(56, 57, 62, 59, 58, 60)$ and $S''_2=(63, 64, 68, 67, 66, 65)$ whose OP patterns are $p''_1=(1, 2, 6, 4, 3, 5)$ and $p''_2=(1, 2, 6, 5, 4, 3)$, respectively, and have frequency $1 < \tau$.\label{example2}
\end{example}
 
Although \MOPPM utilizes efficient methods for scanning time series, its time complexity is in $\Omega(n^3)$ in the worst case, and it \emph{does not scale to large time series} in practice.

Our goal here is to propose exact, highly scalable algorithms for FPM in the OP setting. We would like to have mining algorithms with \emph{both} theoretical guarantees (near-linear time and linear space) and practically efficient implementations. These characteristics are of utmost importance, as time series are ever-increasing in size. The two key ideas of our approach are to: (I) 
develop a time- and space-efficient data structure (\emph{index}) for storing and querying time series; and (II) 
exploit compact representations of OP patterns that take $\cO(n)$ space.

\paragraph{Contributions.}~Our work makes four main contributions. 

\begin{enumerate}
    \item We observe that the \emph{order-preserving suffix tree} (\OPST), introduced in~\cite{DBLP:journals/tcs/CrochemoreIKKLP16}, is essential to designing efficient algorithms for  frequent OP pattern mining. Although other OP indexes exist~\cite{DBLP:conf/esa/GagieMV17,DBLP:conf/icalp/0002PST21,DBLP:conf/isaac/KimC22}, they either focus only on \emph{pattern matching} queries and do not have the full suffix tree functionality that is essential for \emph{frequent pattern mining}; or they do not admit an \emph{efficient} construction. We thus design a theoretically and practically efficient algorithm for constructing the \OPST of an input string (time series) $w$ of length $n$ over a totally ordered alphabet of size $\sigma$, 
which takes $\cO(n\sigma\log \sigma)$ time using $\cO(n)$ space. 
We were motivated to design this algorithm because there are no practical algorithms for \OPST construction; the state-of-the-art  algorithm of~\cite{DBLP:journals/tcs/CrochemoreIKKLP16} is of theoretical interest and to the best of our knowledge has no available implementation. Notably, the $\cO(n\sigma\log \sigma)$ time bound is never attained in our extensive experiments, where our implementation always runs in $\cO(n)$ time irrespective of the alphabet size $\sigma$. 
Thus, our $\cO(n\sigma\log \sigma)$-time algorithm is the first practical algorithm for constructing an $\OPST$. Note that beyond pattern mining, it can be used directly for many other OP analysis tasks; see~\cite{DBLP:journals/tcs/CrochemoreIKKLP16}. 

As a bonus, we provide an $\cO(n\log\sigma)$-time and $\cO(n)$-space algorithm, which is mainly of theoretical interest. Interestingly, this construction algorithm is also theoretically faster than the state of the art~\cite{DBLP:journals/tcs/CrochemoreIKKLP16} for small alphabets. Specifically, the algorithm in~\cite{DBLP:journals/tcs/CrochemoreIKKLP16} has time complexity $\cO(n\sqrt{\log n})$, which is larger than that of our algorithm when $\log\sigma=o(\sqrt{\log n})$.

\item We design an exact algorithm for mining   all \emph{maximal} frequent OP patterns of $w$. By definition, every frequent OP pattern is included in a maximal frequent OP pattern; and the total number of all maximal frequent OP patterns is $\cO(n)$. Our algorithm works in $\cO(n)$ time and space provided that the $\OPST$ of $w$ has been constructed. The main idea is to perform a careful bottom-up traversal of $\OPST$ exploiting the fact that the maximality property  
propagates upward in the tree. By utilizing $\OPST$, we avoid the time series scans of the state-of-the-art \MOPPM algorithm~\cite{xindong_opp}, which takes $\Omega(n^3)$ time in the worst case. Also, as the complexities of our algorithm do \emph{not} depend on the frequency threshold, it can mine patterns of very low frequency efficiently, unlike \MOPPM. 

\item We formalize the notion of \emph{closed} frequent OP patterns. Informally, these are all frequent OP patterns that \emph{cannot} be extended to either side and still have the \emph{same} frequency.    
Importantly, every maximal OP pattern is also closed and the total number of all closed frequent OP patterns is $\cO(n)$. For instance, $p$ in Example~\ref{example1} is a closed frequent OP pattern because it cannot be extended to the left or right and still have frequency $\tau=2$, as explained in Example~\ref{example2}. We design an exact algorithm to compute all closed frequent OP patterns of an  input string $w$. 
Here, unfortunately, the closedness property does not propagate upward in the \OPST of $w$. 
Thus, our main idea is to perform a careful bottom-up traversal enhanced by a novel application of lowest common ancestor queries~\cite{DBLP:conf/latin/BenderF00}. Our algorithm works in $\cO(n)$ time using $\cO(n)$ space provided that $\OPST$ has already been constructed.

\item  We present experiments using $6$ real-world, multi-million letter time series showing that our $\cO(n\sigma \log \sigma)$-time \OPST construction algorithm runs in $\cO(n)$ time on these datasets, despite the $\cO(n\sigma \log \sigma)$ bound, and that our frequent pattern mining algorithms are up to orders of magnitude faster than \MOPPM and natural Apriori-like baselines. Finally, we also present a case study on $2$ publicly available datasets and on a proprietary dataset demonstrating the effectiveness of clustering based on maximal frequent OP patterns. 

\end{enumerate}

\section{Preliminaries}

We view a time series as a string $w=w[0\dd n-1]$ of length $n$ over an integer alphabet $\Sigma$ of size $|\Sigma|=\sigma$. This models any totally ordered alphabet (e.g., that of reals), as clearly the letters of $w$ can be mapped onto the range $[0,\sigma)$ after $\cO(n\log\sigma )$-time preprocessing (e.g., using AVL trees).  

By $p(w)$ we define the function that maps $w$ to the rightmost occurrence $j$ of the largest element of $w[0\dd n-2]$ that is at most equal to $w[n-1]$ (i.e., the \emph{predecessor} of $w[n-1]$). If there is no such $j$ (because $w[n-1]$ is the smallest element), then $p(w)=\perp$. More formally, $p(w)$ is the largest 
$j\leq n-2$ 
such that $w[j]=\max\{w[k]: k \leq n-2, w[k]\leq w[n-1]\}$; if there is no such $j$, then $p(w)=\perp$. By $s(w)$ we define the function that maps $w$ to the rightmost occurrence $j$ of the smallest element of $w[0\dd n-2]$ that is at least equal to $w[n-1]$ (i.e., the \emph{successor} of $w[n-1]$). If there is no such $j$, then $s(w)=\perp$. Similarly, $s(w)$ is the largest 
$j\leq n-2$ 
such that $w[j]=\min\{w[k]: k \leq  n-2, w[k]\geq w[n-1]\}$; if there is no such $j$, then $s(w)=\perp$. For implementation purposes, we map the special symbol $\perp$ to integer $-1$.

\begin{example}
Let $w=5~2~6~5~1~4$ of length $n=6$. 
We have $p(w) = 1$ because $w[1]=2$ is the predecessor of $w[5]=4$ in $w[0\dd 4]$. Since we have only one occurrence of $2$, this one is the rightmost. We have $s(w) = 3$ because $w[3]=5$ is the rightmost occurrence of the successor of $w[5]=4$ in $w[0\dd 4]$. We have two occurrences of $5$ ($w[0]$ and $w[3]$); we pick $3$.
\end{example}
We define and use throughout the following two \emph{codes}~\cite{DBLP:journals/tcs/CrochemoreIKKLP16}: 
\begin{align*}
\Lcode(w)=&(p(w),s(w))\text{; and}\\ 
\Pcode(w)=&\Lcode(w[0\dd 0])\cdot \Lcode(w[0\dd 1])\cdot\\
&\ldots\cdot\Lcode(w[0\dd n-1]),
\end{align*}
as the sequence of the $\Lcode$'s of all prefixes of $w$. 

\begin{example}
For $w=4~2~5~5~1$, $\Lcode(w)=(\perp,1)$ and
$\Pcode(w)=(\perp,\perp)(\perp,0)(0,\perp)(2,2)(\perp,1)$.
\end{example}

Two strings $x$ and $y$ of the same length are called \emph{order-preserving} (OP), denoted by $x \approx y$, if and only if the relative order of their letters is the same. More formally, for all $i,j\in[0,|x|)$, it holds that $x[i]\leq x[j] \iff y[i]\leq y[j]$.

\begin{example}\label{bg:ex1}
For $x=4~2~5~5~1$ and $y=5~2~7~7~0$, $x\approx y$. 
\end{example}

The following lemma from~\cite{DBLP:journals/tcs/CrochemoreIKKLP16} links the OP pattern notion with $\Lcode$ and $\Pcode$. Indeed, we use $\Pcode$ throughout to represent OP patterns.

\begin{lemma}[\cite{DBLP:journals/tcs/CrochemoreIKKLP16}]\label{lem:ois}
Let $x$ and $y$ be two strings. Then
\begin{enumerate}
    \item $x \approx y \iff x[0\dd |x|-2] \approx y[0\dd |y|-2] \textbf{ and } \Lcode(x)=\Lcode(y)$.
    \item $x \approx y \iff \Pcode(x)=\Pcode(y)$.
\end{enumerate}
\end{lemma}

\begin{example}[cont'd from Example~\ref{bg:ex1}]
$\Lcode(x)=\Lcode(y)=(\perp,1)$ and $\Pcode(x)=\Pcode(y)$.
\end{example}

\section{The OP Suffix Tree}

The OP suffix tree of a string $w$ over alphabet $\Sigma$, introduced in~\cite{DBLP:journals/tcs/CrochemoreIKKLP16} and denoted here by $\OPST(w)$, is a compacted trie of the following family $\Scode(w)$ of sequences:
\begin{align*}
\Scode(w)=&\{\Pcode(w[0\dd n-1])\dollar,\\
&\Pcode(w[1\dd n-1])\dollar,\ldots,\\
&\Pcode(w[n-1\dd n-1])\dollar\},\end{align*}

\noindent where each sequence is the \Pcode of a suffix of $w$ and $\$\notin\Sigma$ is a delimiter; inspect \cref{fig:OIST} for an example. 

$\OPST(w)$ has exactly $n$ leaf nodes. The internal nodes of $\OPST(w)$ with at least two children are called \emph{branching} nodes (e.g., $v_4$ in \cref{fig:OIST}). Every branching node $v$ represents $\Pcode(w[i\dd j])$ and stores a pointer $\Slink(v)=u$, called \emph{suffix-link}, to the node $u$ representing $\Pcode(w[i+1\dd j])$ (e.g., $\Slink(v_4)=v_1$ in \cref{fig:OIST}); node $u$ is called \emph{suffix-link node} and it may be branching or non-branching. The root node, the branching nodes, the non-branching suffix-link nodes, and the leaf nodes form the set of \emph{explicit} nodes. 
All the remaining nodes (dissolved in the compacted trie) are called \emph{implicit}. 

For an implicit node $v$, its explicit descendant, denoted by $\Fdown(v)$, is the top-most explicit node in the subtree of $v$.  If $v$ is explicit, $\Fdown(v)=v$.
The \emph{string depth} of a node $v$ is the length of $\Pcode(x)$, where $x$ is a substring of $w$ represented by node $v$ (e.g., it is $3$ for $v_4$ in \cref{fig:OIST}). The \emph{locus} of a node $v$ is defined as the pair $(\Fdown(v),d)$, where $d$ is the string depth of $v$; the root node represents the empty string of length $0$. 
By $\Locus(x)$, we denote a function that outputs the locus of a node corresponding to substring $x$ of $w$. 
In \cref{fig:OIST}, 
$\Locus(5~5~1)=\Locus(4~4~2)=(v_4,3)$. \cref{lem:locus} is a direct implication of \cref{lem:ois} and the $\OPST(w)$ construction.

\begin{lemma}[\cite{DBLP:journals/tcs/CrochemoreIKKLP16}]\label{lem:locus}
Two substrings $x$ and $y$ of $w$ are OP if and only if $\Locus(x)=\Locus(y)$ in $\OPST(w)$.
\end{lemma}

Each explicit node $u$ stores its string depth $\Depth(u)$ and a witness occurrence (position) $\Witness(u)$ of a substring it represents. In particular, every leaf node corresponding to $\Pcode(w[i\dd n-1])\dollar$ stores $i$ as its witness. In addition, each explicit node stores a (possibly empty) set of outgoing edges. Consider one such edge $(u,v)$. Let $i=\Witness(v)$ and $d=\Depth(u)$. The edge $(u,v)$ stores the code $\Lcode(w[i\dd i+d])$, which we call \emph{edge label}. The set of edges outgoing from a node are sorted by the standard lexicographic order assuming that $\texttt{\$}$ $<$ $\perp$ and $\perp$ $<$ $i$, for every $i\in[0,n-2]$. 
Only explicit nodes and their outgoing edges are stored. The tree has $\Theta(n)$ explicit nodes and edges. 

\begin{lemma}[\cite{DBLP:journals/tcs/CrochemoreIKKLP16}]
For any string $w$ of length $n$, the total size of $\OPST(w)$ is $\Theta(n)$.\label{lem:OPSTsize}
\end{lemma}

\begin{figure}[t]
    \begin{subfigure}{0.2\textwidth}
    {\scriptsize $0:1~2~4~4~2~5~5~1\vphantom{)}$}\\
    {\scriptsize $1:2~4~4~2~5~5~1\vphantom{)}$}\\
    {\scriptsize $2:4~4~2~5~5~1\vphantom{)}$}\\
    {\scriptsize $3:4~2~5~5~1\vphantom{)}$}\\
    {\scriptsize $4:2~5~5~1\vphantom{)}$}\\
    {\scriptsize $5:5~5~1\vphantom{)}$}\\
    {\scriptsize $6:5~1\vphantom{)}$}\\ 
    {\scriptsize $7:1\vphantom{)}$}    
    \caption{Suffixes of $w$.}
    \end{subfigure}
    \begin{subfigure}{0.4\textwidth} \label{fig:Scode}
    {\scriptsize $(\perp,\perp)~(0,\perp)~(1,\perp)~(2,2)~(1,1)~(3,\perp)~(5,5)~(0,0)$~\texttt{\$}}\\
    {\scriptsize $(\perp,\perp)~(0,\perp)~(1,1)~(0,0)~(2,\perp)~(4,4)~(\perp,3)$~\texttt{\$}}\\
    {\scriptsize $(\perp,\perp)~(0,0)~(\perp,1)~(1,\perp)~(3,3)~(\perp,2)$~\texttt{\$}}\\
    {\scriptsize $(\perp,\perp)~(\perp,0)~(0,\perp)~(2,2)~(\perp,1)$~\texttt{\$}}\\
    {\scriptsize $(\perp,\perp)~(0,\perp)~(1,1)~(\perp,0)$~\texttt{\$}}\\
    {\scriptsize $(\perp,\perp)~(0,0)~(\perp,1)$~\texttt{\$}}\\
    {\scriptsize $(\perp,\perp)~(\perp,0)$~\texttt{\$}}\\
    {\scriptsize $(\perp,\perp)$~\texttt{\$}}
    \caption{\textsf{SufCode}$(w)$.} 
    \end{subfigure}
    \begin{subfigure}{0.38\textwidth}
    \centering
    \includegraphics[width=5cm]{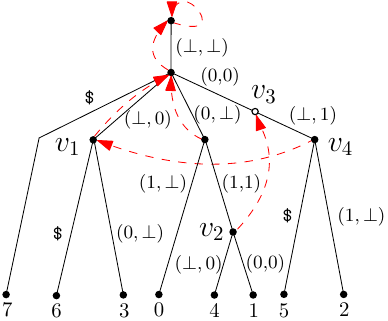}
    \caption{$\OPST(w)$ with suffix links in red. The information stored in explicit nodes is not shown. The node $v_3=\textsf{\small SufLink}(v_2)$ representing $(\perp,\perp)~(0,0)$ is an explicit non-branching suffix-link node.} 
    \label{fig:OIST_slinks}
    \end{subfigure}
    \caption{Suffixes of $w$, $\Scode(w)$, and $\OPST(w)$ for string $w=1~2~4~4~2~5~5~1$ over alphabet $\Sigma=[1,5]$ of size $\sigma=5$.}
    \label{fig:OIST}
\end{figure}

\begin{example}
\cref{fig:OIST} depicts the $\OPST(w)$ for $w=1~2~4~4~2~5~5~1$. Node $v_2$ is branching. It  stores its string depth, which is $3$, and represents $(\perp,\perp)~(0,\perp)~(1,1)$.
One witness of $v_2$ is position $1$ representing substring $x =w[1
\dd 3]=2~4~4$, and
another is position $4$ representing substring $y = w[4\dd 6]=2~5~5 \approx x$. $\Slink(v_2)$ points to the explicit non-branching node $v_3$ that represents $(\perp,\perp)~(0,0)$.
Indeed, $\Pcode(2~4~4)=\Pcode(2~5~5) \implies \Pcode(4~4)=\Pcode(5~5)=(\perp,\perp)~(0,0)$.
\end{example}

\section{The OPST Construction Algorithm}

We present an $\cO(n\sigma\log\sigma)$-time and $\cO(n)$-space algorithm to construct $\OPST(w)$ for string $w$. Before that we show: the preprocessing of $w$ into a compact data structure for computing the $\Lcode$ of any given fragment of $w$; and two crucial combinatorial lemmas: the first one bounds the outdegree of any branching node of $\OPST(w)$; and the second the number of explicit non-branching nodes in any root-to-leaf path.
In the end, we provide a faster alternative algorithm to construct $\OPST(w)$ in $\cO(n\log\sigma)$ time using $\cO(n)$ space.

\paragraph{Preprocessing.}~We preprocess $w$ to efficiently compute the \Lcode of $w[i\dd j]$, for any range $[i,j]$. In particular, we preprocess $w$ for \emph{range predecessor/successor search}~\cite{DBLP:conf/soda/BabenkoGKS15} to directly obtain \cref{lem:oracle}, which we term the \emph{letter oracle}.

\begin{lemma}[Letter oracle]\label{lem:oracle}
Any string $w$ of length $n$ over an alphabet $[0,\sigma)$, with $\sigma\leq n$, can be preprocessed in 
$\cO(n\log\sigma/\sqrt{\log n})$ time in a data structure of $\cO(n \log \sigma/ \log n)$ size so that, given two integers $i\leq j$ from $[0,n)$, we can return $\Lcode(w[i\dd j])$ in $\cO(\log \sigma)$ time.
\end{lemma}

\paragraph{Combinatorial Results.}~While, by Lemma~\ref{lem:OPSTsize}, $\OPST(w)$ contains $\Theta(n)$ explicit nodes and edges, the maximal outdegree of any branching node of the tree is not known. For our construction algorithm, this is crucial. We next show a lemma bounding the outdegree of any branching node of $\OPST(w)$ by $\cO(\sigma)$. This allows us to navigate through $\OPST(w)$ efficiently. 

\begin{lemma}\label{lem:outdegree}
Any branching node in $\OPST(w)$, for any string $w$, has maximal outdegree $2\sigma + 1=\cO(\sigma)$.    
\end{lemma}
\begin{proof}
    We first show that, 
    for any string $w$ over a totally ordered alphabet of size $\sigma$, there exist only up to $2\sigma + 1$ possibilities for $\Lcode(w\cdot a)$ over all letters $a$.
    
    Let us start with the trivial case, where $w$ is the empty string and so $\Lcode(w\cdot a)=\Lcode(a)=(\perp,\perp)$, for any $a$.
    In this case, the above upper bound clearly holds. 
    
    For the rest of the proof, we assume that $w$ is nonempty. We have the following two cases:
    \begin{enumerate}
        \item If letter $a$ belongs to the alphabet of $w$, then $\Lcode(w\cdot a)$ is equal to $(x_a,x_a)$, where $x_a$ is the \emph{single position} of the rightmost occurrence of $a$ in $w$. These are $\sigma$ possibilities: one for every $a$.
        \item Otherwise let the letters of $w$ in the sorted order be equal to $a_j$ for $j\in[0,\sigma)$, and let $x_j$ be the position of the rightmost occurrence of $a_j$ in $w$.
    There exists a single value $j\in[-1,\sigma)$ such that $a_j<a<a_{j+1}$, where $a_{-1}=-\infty$ and $a_{\sigma}=\infty$.
    Then $\Lcode(w\cdot a)$=$(x_j,x_{j+1})$, where $x_{-1}=x_{\sigma}=\perp$.
    Note that, in particular, only elements $x_j$ can occur in $\Lcode$. 
    These are $\sigma+1$ possibilities: $(\perp,x_0),(x_0,x_1),(x_1,x_2),\ldots,(x_{\sigma-1},\perp)$.
    \end{enumerate}   
    
    The upper bound of $2\sigma + 1$ follows by combining cases 1 and 2 above.
    By the above upper bound, any branching node of $\OPST$ cannot have more than
    $2\sigma + 1$ outgoing edges.
\end{proof}

In the standard suffix tree, every internal explicit node is branching. In $\OPST(w)$, however, we may have an internal explicit node (other than the root) which is non-branching (e.g., $v_3$ in \cref{fig:OIST_slinks}). \cref{lem:num_of_expl} bounds the number of explicit non-branching nodes on any root-to-leaf path; this is crucial to prove the time complexity of our construction algorithm.

\begin{lemma}\label{lem:num_of_expl}
There can be at most $\sigma$ explicit non-branching nodes on any  root-to-leaf path in $\OPST(w)$, and in particular between any two branching nodes. 
\end{lemma}

\begin{proof}
We first prove that an explicit non-branching node in $\OPST(w)$ can appear on a root-to-leaf path with witness suffix $w'$ only right before the first occurrence of some letter.

Let $xb$ be a witness substring of a node directly below an explicit non-branching node $u$, where $x$ is a string and $b$ is a letter.
For the explicit non-branching node to be constructed there must exist a branching node $v$ such that $\Slink(v)=u$. Let the witnesses of any two children of $v$ be equal to $a'x'b'$ and $a''x''b''$, respectively, where $x',x''$ are strings and $a',a'',b',b''$ are letters.
We have:
\begin{itemize}
\item $a'x'\approx a''x''$, since both those witnesses are represented by the same node $v$.
\item $a'x'b' \not\approx a''x''b''$, since the two children of $v$ are distinct nodes -- recall that node $v$ is branching.
\item $x'b'\approx xb\approx x''b''$, since $\Locus$ of both $x'b'$ and $x''b''$ (witnesses of children of $v$ with their first letter deleted) must be equal to the single child of $u$ -- recall that node $u$ is explicit but non-branching so it has only one child.
\end{itemize}

Notice that for two equal-length strings $w_1\not\approx w_2$, there must exist two positions $i,j$ such that the relationship ($<$, $=$, $>$) between $w_1[i]$ and $w_1[j]$ is different than the relationship between $w_2[i]$ and $w_2[j]$. Hence the relationship between $a'$ and $b'$ has to be different than the relationship between $a''$ and $b''$ as the relationship of all the other pairs of the letters of $a'x'b'$ and $a''x''b''$ is fixed by the two listed equivalences.
This proves that $b'$ cannot appear in $x'$ as if that was the case the relationship between $a'$ and $b'$ (as a letter of $x'$) would be fixed and the same as the relationship between $a''$ and $b''$ -- a contradiction.
Thus $b$ does not appear in $x$ since $xb\approx x'b'$.

Now, since the explicit non-branching node on a root-to-leaf path can only appear just above the first occurrence of a letter, there are only at most $\sigma$ such nodes. Clearly, this also implies that there can only be at most $\sigma$ explicit non-branching nodes 
between any two branching nodes. 
\end{proof}

\noindent{\bf The Algorithm.}~\cref{alg:con} presents our construction algorithm for $\OPST(w)$. The main idea is to add the sequences of $\Scode(w)$ to an initially empty  compacted trie one at a time, starting from the longest one. 
If we remove the lines colored red in~\cref{alg:con}, we obtain the bruteforce algorithm which inserts every sequence of $\Scode(w)$ in the trie from the root. We next explain the bruteforce algorithm. 

In Lines \ref{alg:con:line1}-\ref{alg:con:line2}, the bruteforce algorithm initializes the trie to have only the root. The active node $u$ and the active string depth $d$ are set to the root and to $0$, respectively (Line~\ref{alg:con:line3}).
The \emph{for} loop in Line~\ref{alg:con:line5} inserts all sequences of $\Scode(w)$ to the trie.
The \emph{while} loop in Line~\ref{alg:con:while1} traverses explicit nodes
using the $\textsf{\small Child}$ function, which returns the appropriate child of $u$, and updates $(u,d)$ upon success (Line~\ref{alg:con:while_suc}).    
The \emph{while} loop in  Line~\ref{alg:con:while2} traverses implicit nodes
by comparing the $\textsf{\small LastCode}$'s and increments $d$ upon success (Line~\ref{alg:con:while2_suc}). We create a branching node $u$ at depth $d$ by \textsf{\small CreateNode}, which also stores as $\Witness(u)$ the smallest $i$ such that the locus of $w[i\dd i+d-1]$ 
is $(u,d)$ (Lines~\ref{alg:con:line10}-\ref{alg:con:line11}).  
We can do this as we insert the sequences of $\Scode(w)$ from the longest to the shortest. Finally, we create the leaf node labeled $i$ (Line~\ref{alg:con:line12}).

We now explain how our algorithm becomes efficient using \emph{suffix links} of branching nodes: these are created during the construction algorithm (see~\cref{alg:slink-simple}) and used as shortcuts for subsequent additions of sequences of $\Scode(w)$. A suffix link is a \emph{shortcut}, as $\Slink(\Locus(x))$ takes us directly to $\Locus(x[1\dd |x|-1])$, for any substring $x$ of $w$, instead of passing always via the root node. To compute the suffix link of a branching node $u$ efficiently, we employ \cref{lem:num_of_expl}.

\begin{algorithm}
\caption{Construction Algorithm for $\OPST(w)$}\label{alg:con}
\begin{algorithmic}[1]\footnotesize
    \Require The letter oracle for $w[0\dd n-1]$ given by \cref{lem:oracle}. 
	\State $root\gets \textsf{\footnotesize EmptyNode}()$\label{alg:con:line1}
	\State $\textsf{\footnotesize Depth}(root)\gets 0$\label{alg:con:line2}
	\State $(u,d)\gets (root,0)$\label{alg:con:line3} 
	\State \textcolor{red}{$\textsf{\footnotesize SufLink}(root) \gets root$}\label{alg:con:line4}
	\For{$i\in[0,n-1]$}\label{alg:con:line5}
	  \While{$d = \textsf{\footnotesize Depth}(u)$ \textbf{ and }
   $\textsf{\footnotesize Child}(u, \textsf{\footnotesize LastCodeInt}(w[i\dd i +d]))$ \text{is} \textit{not empty}}\label{alg:con:while1} 
	  \State $(u,d) \gets (\textsf{\footnotesize Child}(u, \textsf{\footnotesize LastCodeInt}(w[i\dd i + d])),d+1)$\label{alg:con:while_suc}
\While{$d < \textsf{\footnotesize Depth}(u)$ \textbf{ and } $\textsf{\footnotesize LastCode}(w[\textsf{\footnotesize Witness}(u)\dd \textsf{\footnotesize Witness}(u) + d]) = \textsf{\footnotesize LastCode}(w[i\dd i + d])$)}\label{alg:con:while2} 
	    \State $d \gets d+1$\label{alg:con:while2_suc}
	    \EndWhile
	  \EndWhile
	  \If {$d < \textsf{\footnotesize Depth}(u)$}\label{alg:con:line10}
	   \State $u \gets \textsf{\footnotesize CreateNode}(u,d)$\label{alg:con:line11}
	  \EndIf
	  \State $\textsf{CreateLeaf}(i,u,d)$\label{alg:con:line12}
	  \If{\textcolor{red}{$\textsf{\footnotesize SufLink}(u) \textit{ is empty }$}} 
	   \State \textcolor{red}{$\textsf{\footnotesize SufLink}(u)\gets \textsf{ComputeSuffixLinkSimple}(u)$}
	  \EndIf
    \State \textcolor{red}{$(u,d)\gets (\textsf{\footnotesize SufLink}(u),\max\{d-1,0\})$}
	\EndFor
\end{algorithmic}
\end{algorithm}

\cref{alg:con} is a non-trivial  
adaptation of the algorithm by McCreight for standard suffix tree construction~\cite{DBLP:journals/jacm/McCreight76}.
The key differences between \cref{alg:con} and~\cite{DBLP:journals/jacm/McCreight76}  
are explained next.

\emph{How we traverse explicit nodes}: Let $i$ be the witness occurrence stored by some explicit node $u$. Recall that edge $(u,v)$ stores $\Lcode(w[i\dd i+d])$. Instead of this, we map $\Lcode(w[i\dd i+d])=(a,b)\mapsto \LcodeInt(w[i\dd i+d])=(a+1)(n+1)+b+1$, where $\perp$ is treated as $-1$, for every edge in the tree, and store 
$\LcodeInt(w[i\dd i+d])$, 
so that the resulting values are integers in the range $[0,(n+1)^2]$.
It should then be clear that the following property holds for this mapping:
$(a_1,b_1)\leq(a_2,b_2) \iff (a_1+1)(n+1)+b_1+1 \leq (a_2+1)(n+1)+b_2+1$.
By \cref{lem:outdegree}, accessing an edge from a node $u$ by the edge label (\LcodeInt) can be achieved in $\cO(\log \sigma)$ time using AVL trees. This is because the mapping preserves the total order, and we have $\cO(\sigma)$ outgoing edges per node $u$. This transition from $u$ to $v$ is implemented by the \textsf{\small Child} function. It can also be performed in $\cO(1)$ time \emph{with high probability} (whp) via employing perfect hashing~\cite{DBLP:journals/jacm/BenderCFKT23}. 

\emph{How we traverse implicit nodes}: In the standard suffix tree, we have random access to string $w$, and so comparing any two letters takes $\cO(1)$ time merely by accessing $w$. Here we use the letter oracle given by \cref{lem:oracle} to compute the $\Lcode$ of different fragments of $w$. This gives an $\cO(\log\sigma)$-time comparison between any two such $\textsf{\small LastCode}$'s.

\emph{How we compute suffix links (\cref{alg:slink-simple})}: Unlike the algorithm by McCreight, we may construct an explicit node during the suffix-link construction which never becomes branching by the end of the construction; e.g., see $v_3$ in \cref{fig:OIST_slinks}. \cref{alg:slink-simple} presents a simple approach for computing the suffix link of a node $u$. It mimics McCreight's algorithm in that it first finds the nearest ancestor $u'$ of $u$ that has a suffix link (Lines \ref{alg:compsuf:line3} and \ref{alg:compsuf:line4}); since $u'$ has its suffix link constructed, we follow it to compute the suffix link of $u$. Unlike McCreight's algorithm though, the parent of $u$ may not have a suffix link, and so locating $u'$ may take more than $\cO(1)$ time. After we locate $u'$, we follow its suffix link which takes us to another node $v$ (Line~\ref{alg:compsuf:line5}), and then traverse down using the letter oracle until we arrive at string depth $d-1$ (Lines~\ref{alg:compsuf:line6} and~\ref{alg:compsuf:line7}). If the node at this depth is implicit (Line~\ref{alg:compsuf:line8}),  we make it explicit (Line~\ref{alg:compsuf:line9}). Finally, we return the latter node (Line~\ref{alg:compsuf:line10}).

\cref{alg:slink-simple} is efficient when the total number of explicit nodes we traverse is $\cO(n)$.  \cref{lem:num_of_expl} bounds the number of explicit non-branching nodes we traverse in Lines \ref{alg:compsuf:line3} and \ref{alg:compsuf:line4} and Lines \ref{alg:compsuf:line6} and \ref{alg:compsuf:line7}, for the computation of the suffix link of one branching node. 
If the total number of explicit non-branching nodes traversed throughout \OPST construction is $k(w)$, the time complexity is $\cO((n+k(w))\log\sigma)$. 
In McCreight's algorithm, $k(w)=\cO(n)$, and this is why it runs in $\cO(n\log\sigma)$ time.
\cref{lem:nsigma} bounds $k(w)$ in the OP setting. 

\begin{algorithm}
\caption{$\textsf{\small ComputeSuffixLinkSimple}(u)$}\label{alg:slink-simple}
\begin{algorithmic}[1]\footnotesize
    \Require The letter oracle for $w[0\dd n-1]$ given by \cref{lem:oracle}.
	\State $d \gets \textsf{\footnotesize Depth}(u)$\label{alg:compsuf:line1}
    \State $u' \gets u$\label{alg:compsuf:line2}
    \While{$\textsf{\footnotesize SufLink}(\textsf{\footnotesize Parent}(u'))$ \textit{empty}}\label{alg:compsuf:line3}
    \State $u' \gets \textsf{\footnotesize Parent}(u')$\label{alg:compsuf:line4}
    \EndWhile
	\State $v \gets \textsf{\footnotesize SufLink}(\textsf{\footnotesize Parent}(u'))$\label{alg:compsuf:line5}
	  \While{$\textsf{\footnotesize Depth}(v) < d-1$}\label{alg:compsuf:line6}  
	  \State $v \gets \textsf{\footnotesize Child}(v, \textsf{\footnotesize LastcodeInt}(w[\textsf{\footnotesize Witness}(u)+ 
    1\dd \textsf{\footnotesize Witness}(u) + 1 + \textsf{\footnotesize Depth}(v)]))$\label{alg:compsuf:line7}
	  \EndWhile  
	  \If{$\textsf{\footnotesize Depth}(v) > d-1$}\label{alg:compsuf:line8}                         
        \State $v \gets \textsf{\footnotesize CreateNode}(v,d-1)$\label{alg:compsuf:line9}  
    \EndIf
    \State {\bf return} $v$\label{alg:compsuf:line10}
\end{algorithmic}
\end{algorithm}

\begin{lemma}\label{lem:nsigma}
In $\OPST(w)$, for any string $w$ of length $n$ over an alphabet of size $\sigma$, we have $k(w)= \cO(n\sigma)$.
\end{lemma}

\begin{proof}
We mimic the proof of the complexity of McCreight's algorithm~\cite{DBLP:journals/jacm/McCreight76}. Therein, the bound on the number $k(w)$ of moves was proved based on the \emph{node depth} of the nodes $u$ and $v$: when going from $u$ to its branching parent $u'$, the node depth decreases by $1$, and then upon using the suffix link of $u'$ it can again decrease by at most $1$. Thus we can do at most $2n$ moves upward, and since the node depth is bounded by $n$ (the maximal node depth in the tree), at most $3n$ moves downward.

In the OP setting, by \cref{lem:num_of_expl} and the above analysis of  McCreight's algorithm, the node depth can decrease by at most $\sigma +2$ when going upwards (which also bounds the possible number of moves) due to the explicit non-branching nodes that do not correspond to any explicit or branching node on the path from root to $v=\Slink(u)$. This means that we can make at most $n(\sigma+2)$ moves upward and at most $n(\sigma+2)+n$ moves downward; that is,  $k(w)\le n(2\sigma+5)=\cO(n\sigma)$.
\end{proof}

Since the time spent at each node in \cref{alg:slink-simple} is $\cO(\log\sigma)$, \cref{lem:nsigma} implies that all invocations of \cref{alg:slink-simple} take $\cO(n\sigma\log\sigma)$ time in total. Since the rest of \cref{alg:con}, like McCreight's algorithm~\cite{DBLP:journals/jacm/McCreight76}, takes $\cO(n\log\sigma)$ time, \cref{alg:con} takes $\cO(n\sigma \log\sigma)$ time and $\cO(n)$ space in total.

\paragraph{Faster Algorithm.}~We next provide a faster algorithm to construct $\OPST(w)$ in $\cO(n\log\sigma)$ time using $\cO(n)$ space. It is mainly of theoretical interest as it relies on a general framework for constructing suffix trees with \emph{missing suffix links}~\cite{DBLP:journals/siamcomp/ColeH03a}. This framework implies an algorithm for constructing $\OPST(w)$ in $\cO(n \cdot \textsf{oracle}(n,\sigma) + \textsf{ORACLE}(n,\sigma))$ time whp, where $\textsf{oracle}$ is the time complexity to compute $\Lcode(w[i\dd j])$, for any $i,j$, and $\textsf{ORACLE}$ is the construction time for this oracle. Using \cref{lem:oracle} gives $\textsf{oracle}(n,\sigma)=\cO(\log\sigma)$ and $\textsf{ORACLE}(n,\sigma)=\cO(n\log\sigma)$, so we obtain an $\cO(n\log\sigma)$-time randomized construction. The source of randomization in~\cite{DBLP:journals/siamcomp/ColeH03a} is an algorithm for navigating through the edges of the tree. With our \cref{lem:outdegree}, this operation takes $\cO(\log\sigma)$ \emph{deterministic} time. Hence we obtain an $\cO(n\log\sigma)$-time $\cO(n)$-space deterministic construction for $\OPST(w)$.

\begin{theorem}
For any string $w$ of length $n$ over a totally ordered alphabet of size $\sigma$, $\OPST(w)$ can be constructed in $\cO(n\log\sigma)$ time using $\cO(n)$ space.
\end{theorem}

\section{Frequent OP Pattern Mining Algorithms}
We show two algorithms for mining OP patterns using $\OPST(w)$. 
We use two standard, compact representations of patterns: \emph{maximal} (\cref{sec:maximal}) and \emph{closed} (\cref{sec:closed}).

\subsection{Mining Maximal Frequent OP Patterns}\label{sec:maximal}

Recall that $w$ is a string of length $n$.
Let $\tau>1$ be an integer.
An OP pattern $p$ is: 
\begin{itemize}
    \item \emph{$\tau$-frequent} in $w$, if there are at least $\tau$ \emph{distinct fragments} 
    $w[i\dd j]$ in $w$ whose OP pattern is $p$; i.e., $p = \Pcode(w[i\dd j])$ for every such fragment. 
    \item \emph{right-$\tau$-maximal} in $w$, if for every fragment $w[i\dd j]$ in $w$ whose OP pattern is $p$, $\Pcode(w[i\dd j+1])$ is \emph{not} $\tau$-frequent in $w$; \emph{left}-$\tau$-\emph{maximality} is defined analogously.
    \item \emph{$\tau$-maximal}, if it is both left- and right-$\tau$-maximal.
\end{itemize}
We may characterize nodes of \OPST based on the properties of the patterns they represent as $\tau$-frequent, $\tau$-maximal, etc.

We next present our $\cO(n)$-time and $\cO(n)$-space algorithm for mining \emph{all $\tau$-maximal $\tau$-frequent OP patterns in $w$}, given an OP suffix tree $\OPST(w)$. The algorithm has two phases. 

\paragraph{Phase I.}~It marks a branching node $v$ of $\OPST(w)$ as a \emph{candidate} node if:
\begin{itemize}
    \item $v$ has at least $\tau$ leaf descendants; and
    \item no descendant of $v$ has at least $\tau$ leaf descendants.
\end{itemize}
We perform this for all branching nodes via a depth-first search (DFS) traversal on $\OPST(w)$. 
The candidate nodes represent all OP patterns that are right-$\tau$-maximal \emph{and} $\tau$-frequent. 

\paragraph{Phase II.}~It identifies  each candidate node $v$ that is \emph{also} left-$\tau$-maximal to obtain
all $\tau$-maximal $\tau$-frequent OP patterns. This is because an OP pattern is $\tau$-maximal and $\tau$-frequent if and only if it labels a candidate node $v$ (so it is right-$\tau$-maximal and $\tau$-frequent) and additionally it is left-$\tau$-maximal. To identify such a $v$, we check whether both of the following conditions hold: (I) $v$ does \emph{not} have an incoming suffix link from another $\tau$-frequent node. Otherwise, the latter node would correspond to an extension of the pattern of $v$ to a $\tau$-frequent pattern. (II) All children of $v$ are left-$\tau$-maximal. Otherwise, a child of $v$ would correspond to an extension of the pattern of $v$ to a $\tau$-frequent pattern. Note that left-$\tau$-maximality is a property that propagates upward: if a branching node $v$ is \emph{not} left-$\tau$-maximal, then \emph{neither} is any of its ancestors. We thus check conditions I and II as follows. 

First, we mark all leaf nodes as left-$\tau$-maximal: when $w[i\dd n-1]$, which corresponds to a leaf, is extended to the left to $w[i-1\dd n-1]$, then $\Pcode(w[i-1\dd n-1])$ has exactly $1<\tau$ occurrence, so the leaf is left-$\tau$-maximal. Then, we perform a DFS traversal on $\OPST(w)$ to 
identify whether there are incoming suffix links (each suffix link is checked only once) from $\tau$-frequent nodes, to check condition I \emph{and} propagate  upward whether a node is left-$\tau$-maximal or not, to check condition II. We output one witness fragment, as an interval on $w$, for every $\tau$-maximal $\tau$-frequent OP pattern we find; in particular, we report $[\Witness(v),\Witness(v)+\Depth(v)-1]$. Since any DFS traversal takes linear time and we have $\cO(n)$ explicit nodes and edges, the total time is $\cO(n)$ provided that $\OPST(w)$ is given. We have thus obtained the following result: 

\begin{theorem}\label{the:maximal}
For any string $w$ of length $n$ and any integer $\tau>1$,
we can list all $\tau$-maximal $\tau$-frequent OP patterns in $w$ 
in $\cO(n)$ time using $\cO(n)$ space if $\OPST(w)$ is given.
\end{theorem}

\begin{example}\label{ex:maximal}
    Let $w=1~2~4~4~2~5~5~1$ (see Fig.~\ref{fig:OIST}) and $\tau=2$. 
    In Phase I, the algorithm marks nodes $v_1$, $v_2$, and $v_4$ as candidates, as they have at least $2$ leaf descendants and no descendant of $v_1$, $v_2$, and $v_4$ has itself at least $2$ leaf descendants. 
   The OP patterns they represent (i.e., $(\perp,\perp)~(\perp, 0)$; $(\perp,\perp)~(0,\perp)~(1,1)$; and $(\perp,\perp)~(0,0)~(\perp,1)$) are all right-$2$-maximal and $2$-frequent. In Phase II,  
   $v_2$ and $v_4$  
   are further identified as left-$2$-maximal: neither node has an incoming suffix link from another $2$-frequent node (condition I) and  all their children are leaf nodes, which are left-$2$-maximal (condition II).  
  We obtain the following $2$-maximal $2$-frequent OP patterns: (I) $(\perp,\perp)~(0,\perp)~(1,1)$, for node $v_2$: it corresponds to the two fragments $2~4~4$, and $2~5~5$ of $w$; and we report interval $[1,3]$ as a witness; and (II) $(\perp,\perp)~(0,0)~(\perp,1)$, for node $v_4$: it corresponds to the two fragments $4~4~2$ and $5~5~1$ of $w$; and we report interval $[2,4]$ as a witness. 
\end{example}

\subsection{Mining Closed Frequent OP Patterns}\label{sec:closed} 
Recall that $w$ is a string of length $n$. An OP pattern $p=\Pcode(x)$ for a string $x$ is:  
\begin{itemize}
    \item \emph{right-closed} in $w$, if there exists a substring $w[i\dd j]\approx x$ 
    whose $\Pcode(w[i\dd j])$ is the $\Pcode$ of strictly more fragments of $w$ than $\Pcode(w[i\dd j+1])$; 
    \emph{left-closedness} is defined analogously. 
    \item \emph{closed}, if it is left-closed and right-closed. 
\end{itemize}
We may characterize nodes of \OPST based on the properties of the patterns they represent, e.g., right-closed, or closed. 

We show an $\cO(n)$-time and $\cO(n)$-space algorithm for mining \emph{all closed $\tau$-frequent OP patterns in $w$} given $\OPST(w)$. 

\paragraph{Preprocessing.}~We preprocess $\OPST(w)$ to answer lowest common ancestor (LCA) queries~\cite{DBLP:conf/latin/BenderF00}  using Lemma~\ref{lem:lca}. These queries are key to the efficiency of our algorithm. 

\begin{lemma}[\cite{DBLP:conf/latin/BenderF00}]\label{lem:lca}
Any rooted tree $T$ on $N$ nodes can be preprocessed in $\cO(N)$ time to support the following $\cO(1)$-time queries: given any two nodes $u$ and $v$ of $T$ return node $\textsf{LCA}(u,v)$, that is, the LCA of $u$ and $v$ in $T$.
\end{lemma}

\paragraph{Phase I.}~It marks each branching node of $\OPST(w)$ with at least $\tau$ leaf descendants as a \emph{candidate} using a DFS traversal. Since OP patterns ending at branching nodes are by construction right-closed, the candidate nodes represent all OP patterns that are 
right-closed \emph{and} $\tau$-frequent.  

\paragraph{Phase II.}~It identifies each candidate that is \emph{also} left-closed, to obtain all closed $\tau$-frequent OP patterns. 

A candidate node $v$ is left-closed if it has at least two witnesses with different codes when extended to the left by a single letter; i.e., $v$ has at least two leaf descendants $i$ and $j$ such that 
$w[i \dd i + \Depth(v) -1] \approx w[j \dd j + \Depth(v) -1]$ but 
$w[i - 1 \dd i + \Depth(v)-1]\not\approx w[j-1\dd j + \Depth(v) -1]$ (recall that $\Depth(v)$ is the string depth of $v$). An OP pattern is closed $\tau$-frequent if and only if it labels a branching node $v$ that is left-closed \emph{and} has at least $\tau$ leaf descendants. 

We next show how left-closed candidates are identified. 
For each candidate node $v$, we need to extend each of its leaf descendants $w[\ell\dd n-1]$ to $w[\ell-1\dd n-1]$ (by a single  letter to the left). Let $L_v$ be the set of the leaf descendants of $v$, each represented by its label. The crucial idea is to check if $v$ is left-closed by computing the string depth of the LCA of leaf nodes $\{\ell-1:\ell\in L_v\}$, denoted by $val(v)$.  
Clearly, if $val(v)<\Depth(v)+1$, then it is because at least two leaf descendants have different codes on their left, and so $v$ is left-closed. This approach avoids the pairwise comparison of the $n$ $\Pcode$'s of the leaf nodes, which would result in $\Omega(n^2)$ time, instead of the $\cO(n)$ time of our approach.  

\begin{example}\label{ex:depth}
    Node $v_4$ in Fig.~\ref{fig:OIST} has $\Depth(v_4) = 3$ and $L_{v_4} = \{5, 2\}$. 
    Since $\{\ell-1: \ell\in L_{v_4}\}=\{4,1 \}$, and 
    $\textsf{LCA}(\{4,1\})$ (i.e., $v_2$) has string depth $3$, we have $val(v_4)=3<\Depth(v_4)+1=3+1=4$. 
    Thus, $v_4$ is left-closed. 
    Node $v_3$ has $ \Depth(v_3) = 2$, $L_{v_3}=\{5,2\}$, $\{\ell-1 : \ell\in L_{v_3}\}=\{4,1\}$, and 
    $\textsf{LCA}(\{4,1\})$ has string depth $3$. Since $val(v_3)=3\geq \Depth(v_3)+1=3$, $v_3$ is \emph{not} left-closed.  
\end{example}

In the standard suffix tree, $val(v)$ is either $\Depth(v)+1$ or $0$: $val(v)=0$ means that $v$ and all its ancestors are left-closed and thus finding closed patterns can be done easily using a DFS. In \OPST, however, $val(v)\in[0, \Depth(v)+1]$, and it implies left-closedness for $v$ \emph{only} if $val(v)<\Depth(v)+1$ (and this propagates only to ancestors that are deep enough). \cref{obs} highlights this essential difference.

\begin{observation}\label{obs}
If $w[i+1\dd j]=w[i'+1\dd j']$ and $w[i\dd j-1]=w[i'\dd j'-1]$ then $w[i\dd j]=w[i'\dd j']$; but $w[i+1\dd j]\approx w[i'+1\dd j']$ and $w[i\dd j-1]\approx w[i'\dd j'-1]$ do not imply $w[i\dd j]\approx w[i'\dd j']$ in the OP setting.
\end{observation}

To compute $val(v)$ for all nodes $v$, instead of explicitly considering $L_v$ for each $v$,  
we use a DFS traversal on $\OPST(w)$ coupled with several LCA queries. For any leaf $v$, $val(v)$ is set to $\Depth(v)+1$. Since $val(v)<\Depth(v)+1$ does not currently hold, every leaf $v$ is not left-closed. For any other node $v$, $val(v)$ is computed as the minimum between the values $val$ of its children and of the string depths of the LCA's of $\Witness(v)-1$ and $\Witness(u)-1$ for every child $u$ of $v$. We compute these LCA's efficiently by employing \cref{lem:lca}: for any set of $k$ nodes $v_1,\ldots,v_k$, we can compute their LCA, denoted by $u_{k-1}$, in $\cO(k)$ time, by computing $u_1=\textsf{LCA}(v_1,v_2),u_2=\textsf{LCA}(u_1,v_3),\ldots, u_{k-1}=\textsf{LCA}(u_{k-2},v_k)$. Then node $v$ is left-closed if and only if $val(v)<\Depth(v)+1$. 
Since any DFS traversal takes linear time and we have $\cO(n)$ explicit nodes and edges, the total time is $\cO(n)$ provided that $\OPST(w)$ is given. We have thus proved the following result:

\begin{theorem}\label{the:closed}
For any string $w$ of length $n$ and any integer $\tau>1$,
we can list all closed $\tau$-frequent OP patterns of $w$ in $\cO(n)$ time using $\cO(n)$ space if $\OPST(w)$ is given.
\end{theorem}

\begin{example}   Let $w=1~2~4~4~2~5~5~1$ (see  Fig.~\ref{fig:OIST}) and $\tau=2$. 
In Phase I, all branching nodes are 
marked as candidates. In Phase II, we obtain $4$ closed $2$-frequent OP patterns: $(\perp,\perp)$;   
$(\perp,\perp)~(0,\perp)$; 
$(\perp,\perp)~(0,\perp)~(1,1)$; and  
$(\perp,\perp)~(0,0)~(\perp,1)$. 
For example, node $v_4$ is left-closed; it is also branching with $2$ leaf descendants. Thus, it represents a closed $2$-frequent OP pattern, namely $(\perp,\perp)~(0,0)~(\perp,1)$. This pattern corresponds to the fragments $4~4~2$ and $5~5~1$ of $w$. On the other hand, node $v_1$ is not left-closed, as $val(v_1)=3=\Depth(v_1)+1$.   
\end{example}

\section{Related Work}\label{sec:related}

Order-preserving indexing~\cite{DBLP:journals/ipl/KubicaKRRW13}
and pattern matching~\cite{DBLP:journals/tcs/CrochemoreIKKLP16} have been studied heavily by the theory community.  However, only very recently practical mining algorithms were  developed. Such algorithms mine various types of OP patterns (frequent~\cite{xindong_opp, tkde_opp}, maximal frequent~\cite{xindong_opp}, co-occurrence~\cite{tkdd_opp}, or 
top-$k$ contrast~\cite{topkopp}),  
or OP rules~\cite{tkde_opp}. The most relevant algorithm to our work is \MOPPM~\cite{xindong_opp}. This algorithm utilizes an efficient pattern matching algorithm, an Apriori-like enumeration strategy, and pruning strategies. Its time and space complexity is $\cO(n\cdot \mu \cdot L)$ and $\cO(\mu\cdot(L+n))$, respectively, where $n$ is the length of the input time series, $\mu$ is the maximum length of a candidate frequent OP pattern, and $L$ is the number of fixed-length candidate frequent OP patterns. Since in the \emph{worst case}, $\mu=\Theta(n)$ and $L=\Theta(n)$, \MOPPM needs $\Omega(n^3)$ time and $\Omega(n^2)$  space. 
Thus, our algorithm for maximal frequent OP patterns is  
faster theoretically and practically; see \cref{the:maximal} and \cref{sec:experiments}, respectively.    

\section{Experimental Evaluation}\label{sec:experiments} 

\paragraph{Datasets.}~We used $5$ publicly available, large-scale datasets from different domains, as well as a proprietary dataset. The proprietary dataset, called \ECG, contains the ECG recordings of 40 participants in a study on exercise-induced pain, which was approved by our Ethics board and performed according to the Declaration of Helsinki. 
The datasets  have different length $n$ and alphabet size $\sigma$ (see Table~\ref{tab:datasets} and \cite{Supplement2024} for details).  
We also used $2$ datasets  
from UCR (available at \url{https://bit.ly/4bCP83m}) that are  
comprised of multiple time series (see \cite{Supplement2024} for details) in our clustering experiments. 

\begin{table}[!t]
\caption{Dataset characteristics}\label{tab:datasets}
\centering
\begin{tabular}{|c||c|c||c|}
\hline
\textbf{Datasets}                   & \bm{$n$}                       &\bm{$\sigma$} & \textbf{Available at} \\ \hline \hline
Household (\HOU)                                & 6,147,840               & 88  & \url{https://bit.ly/4dxHJ75}                                             \\ \hline
Solar (\SOL)                                    & 7,148,160               & 1,666 & \url{https://bit.ly/44y9Wqd}                                              \\ \hline
ECG\_{Pain} (\ECG)                               & 22,973,535              & 31,731 & Proprietary                                          \\ \hline
Traffic (\TRA)                                 & 15,122,928                  & 6,176 & \url{https://bit.ly/4dx4tV4}                                           \\ \hline
Temperature (\TEM)                             & 114,878,720                  & 76 & \url{https://bit.ly/4apQjlT}                                            \\ \hline
Whales (\WHA)                                 & 308,941,605               & 554,635 & \url{https://bit.ly/44ssCbb}                                           \\ \hline
\end{tabular}
\end{table}

\begin{figure*}[!t]
\centering
\begin{subfigure}[t]{0.25\textwidth}
    \centering
    \includegraphics[width=\linewidth]{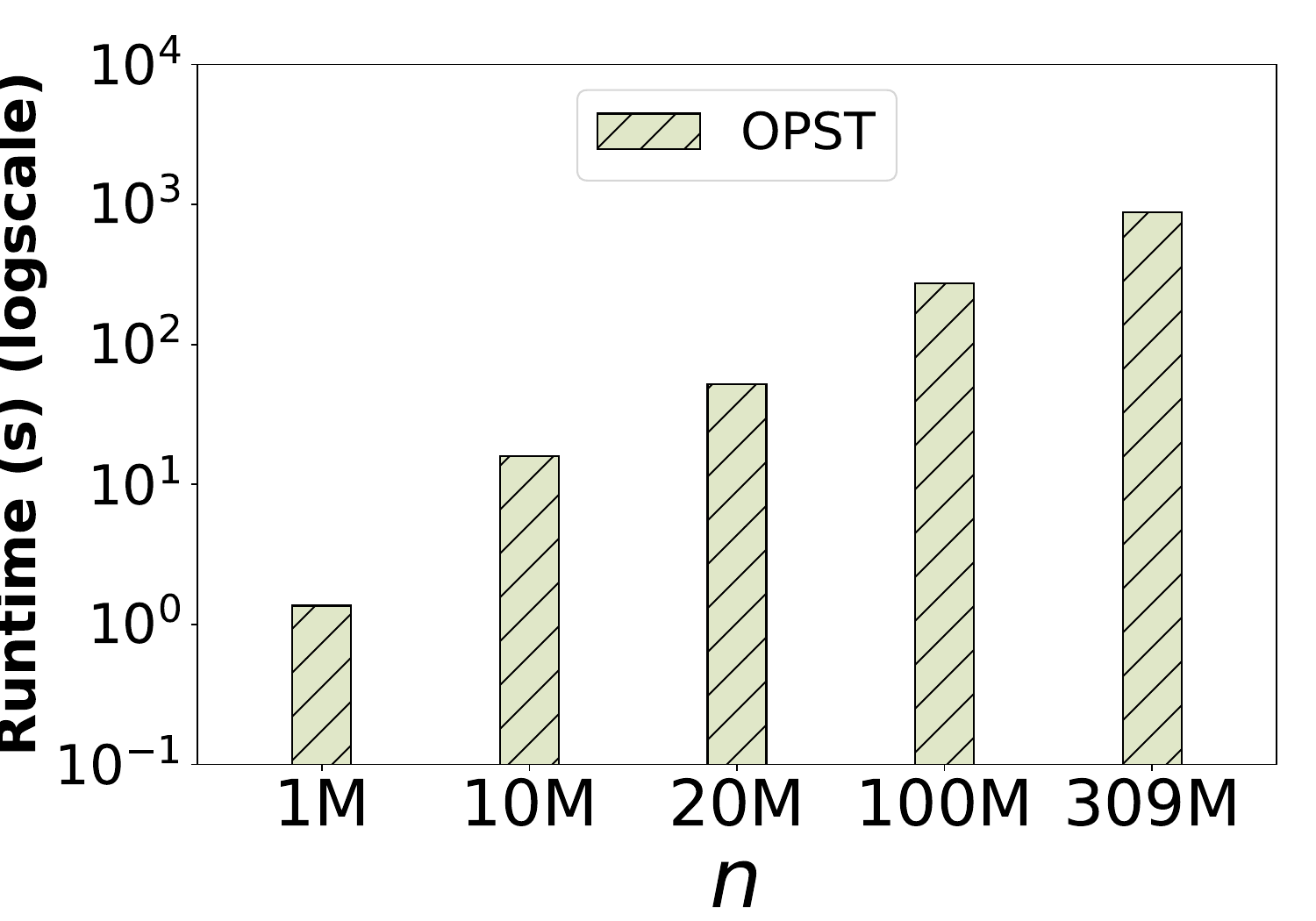}
    \caption{\WHA}\label{fig:runtime_opst_n}
\end{subfigure}%
\hfill
\begin{subfigure}[t]{0.25\textwidth}
    \centering
    \includegraphics[width=\linewidth]{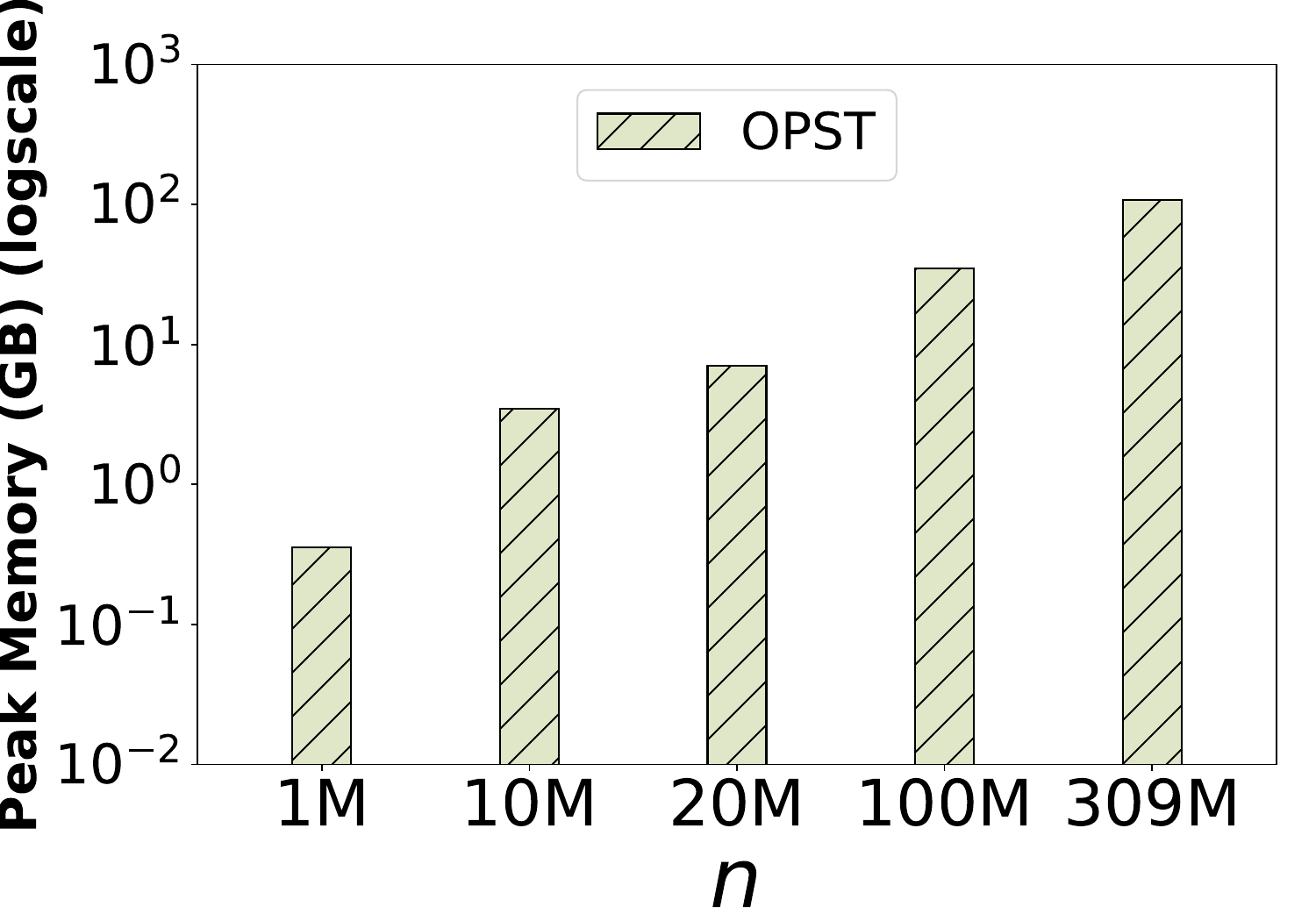}
    \caption{\WHA}\label{fig:mem_opst_n}
\end{subfigure}%
\hfill
\begin{subfigure}[t]{0.25\textwidth}
    \centering
    \includegraphics[width=\linewidth]{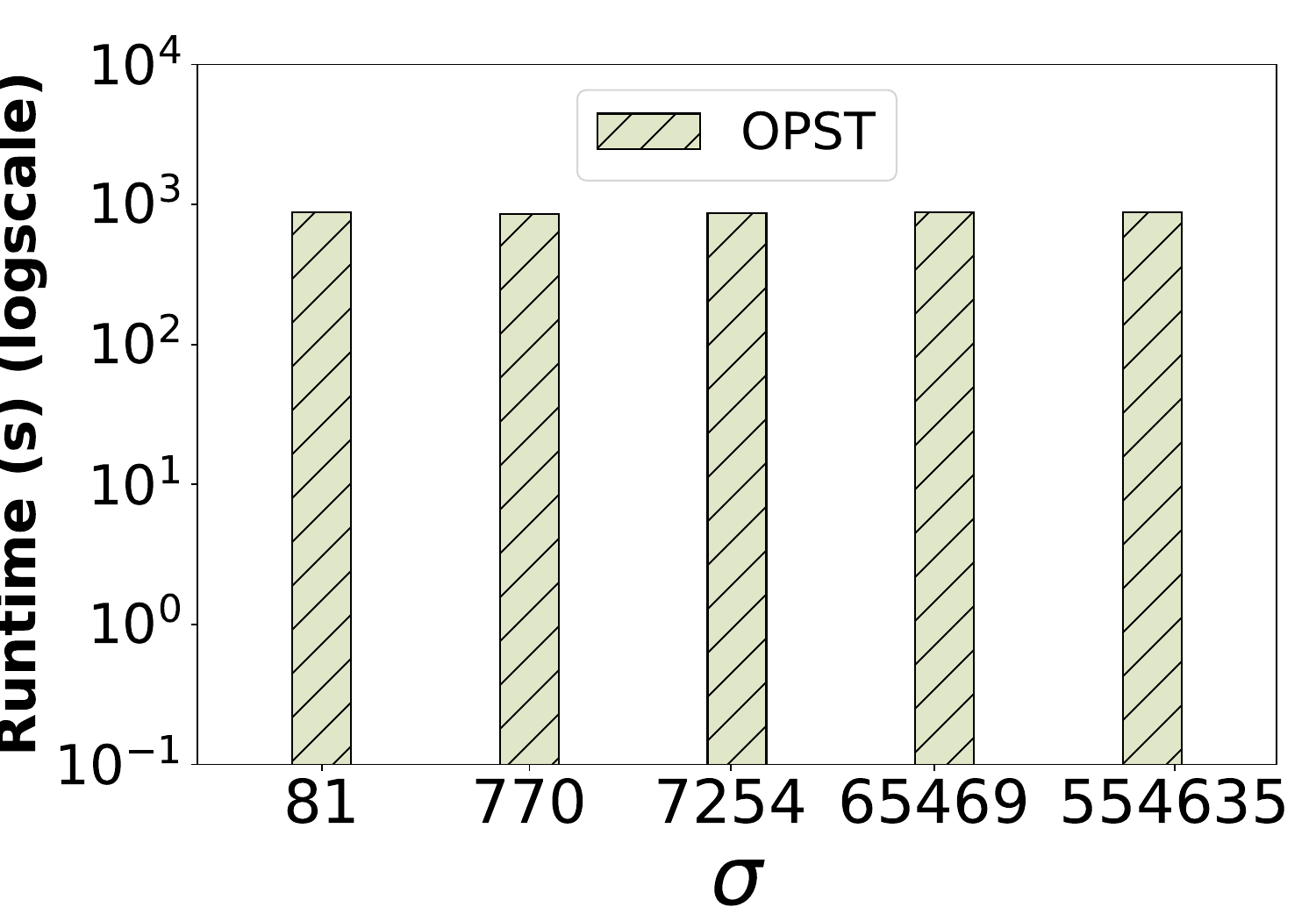}
    \caption{\WHA}\label{fig:runtime_opst_sigma}
\end{subfigure}%
\hfill
\begin{subfigure}[t]{0.25\textwidth}
    \centering
    \includegraphics[width=\linewidth]{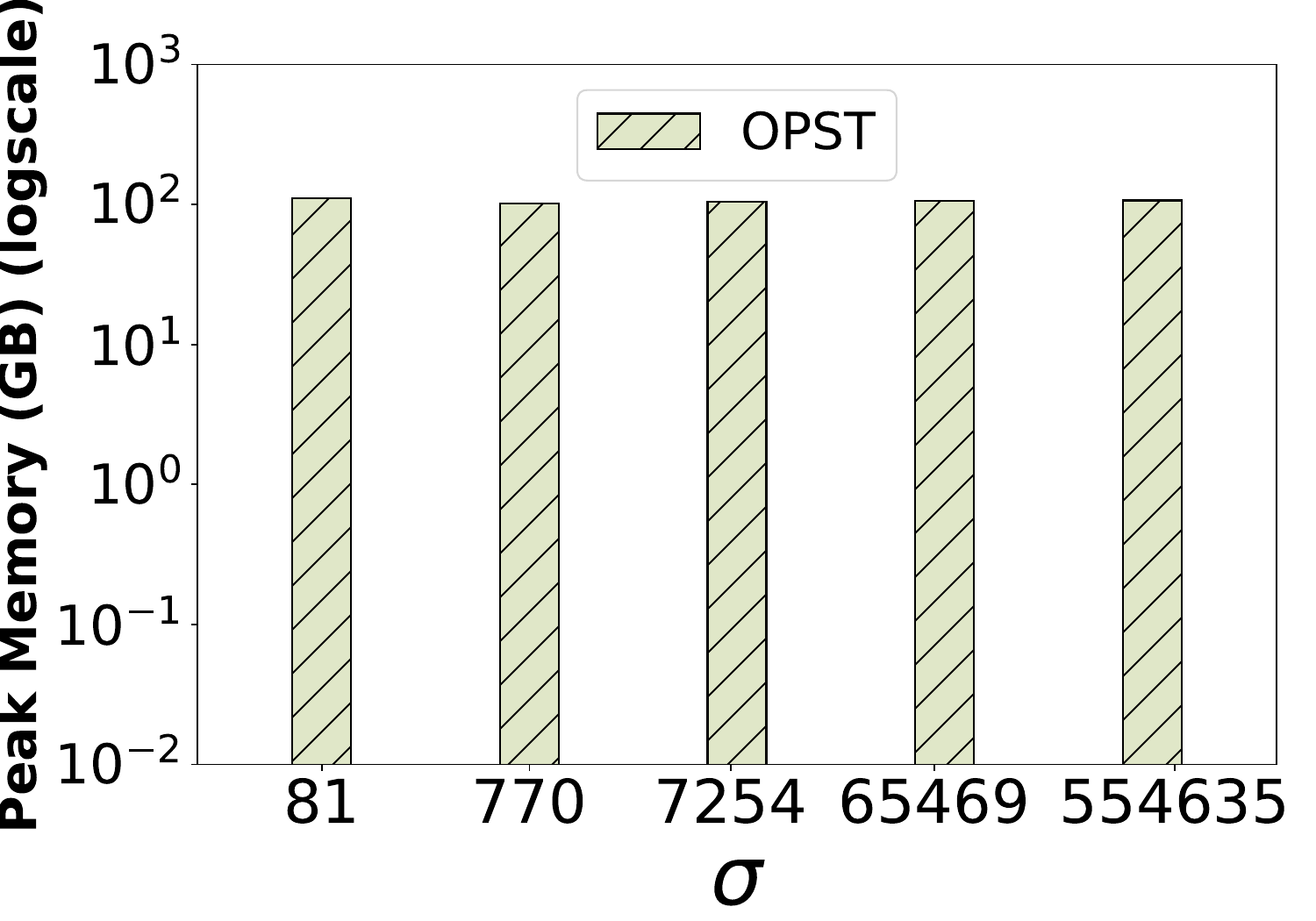}
    \caption{\WHA}\label{fig:mem_opst_sigma}
\end{subfigure}%

\caption{\OPST (construction algorithm): (a) Runtime and (b) peak memory consumption for varying $n$. (c) Runtime and (d) peak memory consumption for varying $\sigma$.}\label{fig:opstWHA}
\end{figure*}

\paragraph{Setup.}~We denote our $\cO(n\sigma \log\sigma)$-time \OPST construction algorithm by \OPST and our algorithm for mining maximal (respectively, closed) frequent OP patterns by \OPSTMP (respectively, \OPSTCP). We compared \OPSTMP to \MOPPM \cite{xindong_opp} (\MOPP) (see Section~\ref{sec:related}) and to a natural Apriori-like~\cite{apriori}  baseline (\BAMP), which takes $\cO(nk\log k)$ time and $\cO(n)$ space, where $k$ is the length of the longest frequent OP pattern output. \BAMP considers increasingly longer substrings and stops when no substring is frequent.   
We compute the $\textsf{\small PrefCode}$'s    
of substrings of $w$ up to length $k$ in $\cO(\log k)$ time per substring using a sliding window technique (see~\cite{Supplement2024} for details). This way we know the frequency of each such $\Pcode$, as well as the frequency of the two $\Pcode$'s corresponding to one letter shorter substrings -- the solutions for \BAMP follow directly. Since no existing algorithm mines closed frequent OP patterns, we compared \OPSTCP to \BACP, an $\cO(nk\log k)$-time and $\cO(n)$-space baseline that is similar to \BAMP (see \cite{Supplement2024} for details). By default, we used $\tau=10$. 

All our experiments ran on an AMD EPYC 7702 
CPU with 1TB RAM. All algorithms are implemented in \texttt{C++}. 
Our source code is available at \url{https://bit.ly/3WssaYf}.

\paragraph{OPST Construction.}~Fig.~\ref{fig:opstWHA} shows the impact of $n$ and $\sigma$ on the runtime and peak memory consumption (construction space) of \OPST. Our algorithm scaled linearly with $n$ in terms of runtime, and its runtime was not affected by $\sigma$ despite the $\cO(n\sigma \log \sigma)$ bound (even when $\sigma$ exceeded $5.5\cdot 10^5$). Our algorithm is efficient; it required about $875$ seconds for the entire \WHA dataset ($n\approx 309\cdot 10^6$). In terms of space, our algorithm scaled linearly with $n$, and as expected by its $\cO(n)$ space complexity, $\sigma$ did not affect its space requirements.   
Similar results 
for all other datasets are reported in \cite{Supplement2024}.

\begin{figure*}[!ht]
\centering 
\begin{subfigure}[t]{0.248\textwidth}
    \centering
    \includegraphics[width=1\linewidth]{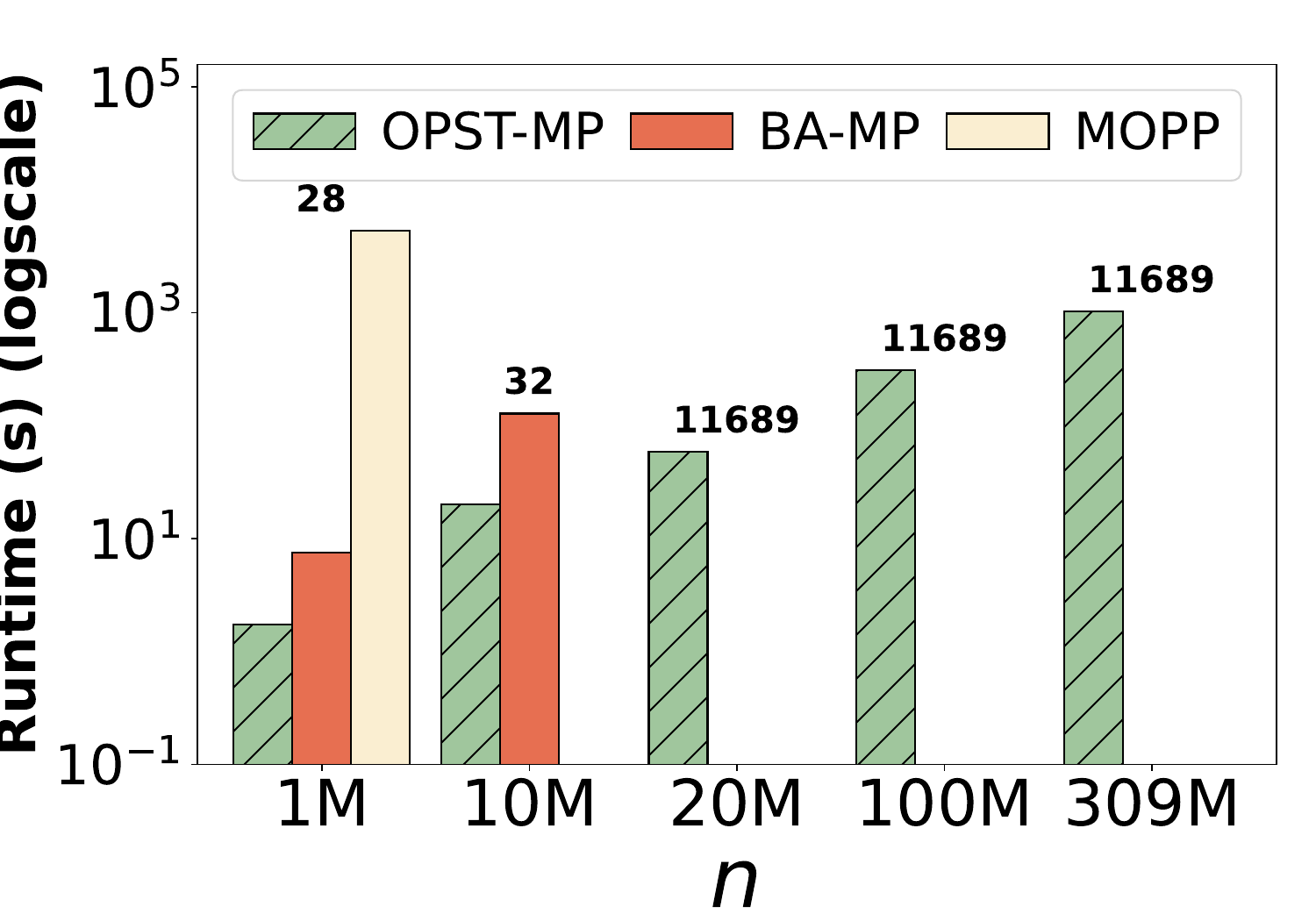}
    \caption{\WHA}\label{fig:runtime_mp_n}
\end{subfigure}%
\hfill
\begin{subfigure}[t]{0.248\textwidth}
    \centering
    \includegraphics[width=1\linewidth]{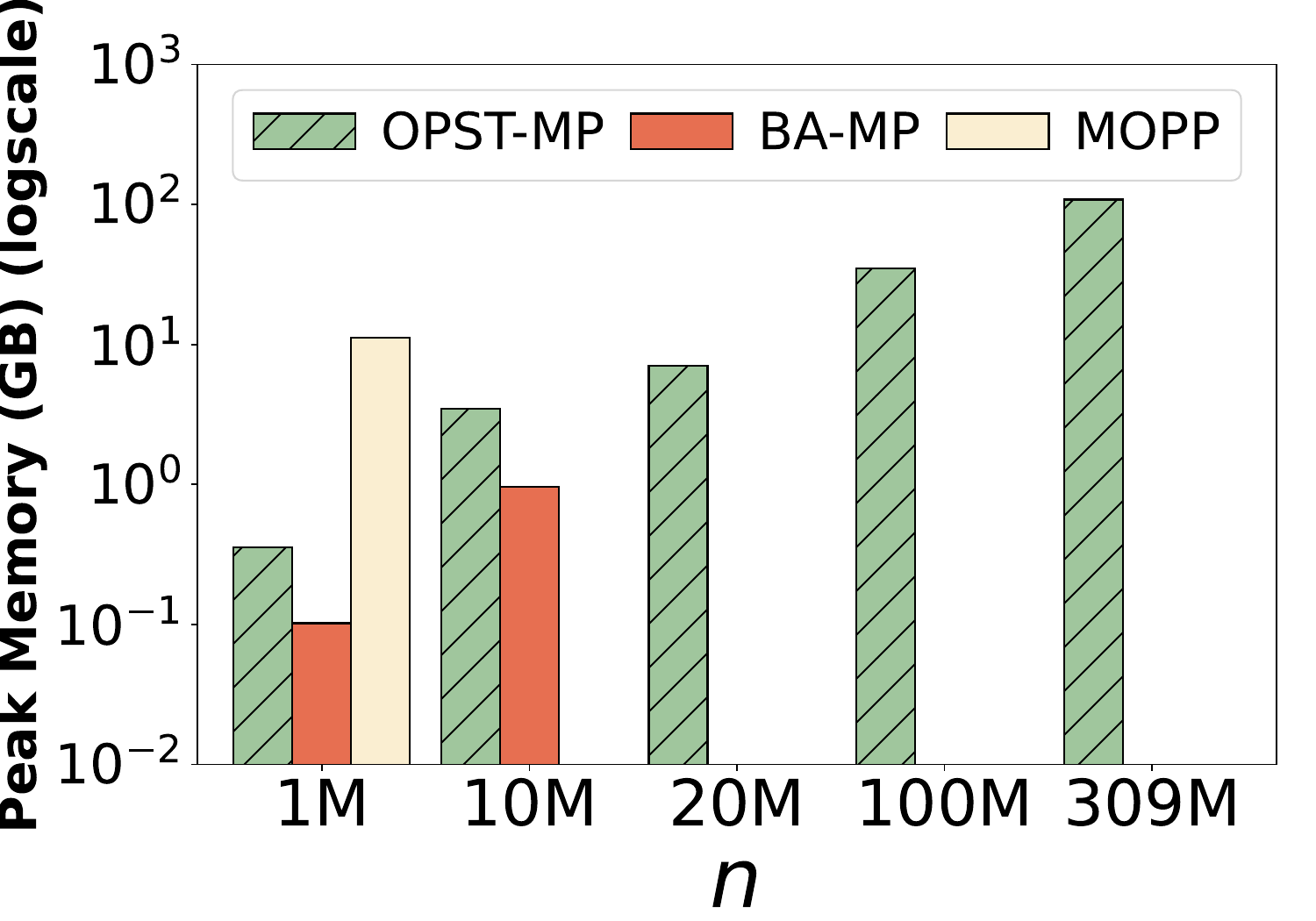}
    \caption{\WHA}\label{fig:mem_mp_n}
\end{subfigure}
\hfill
\begin{subfigure}[t]{0.248\textwidth}
    \centering
    \includegraphics[width=1\linewidth]{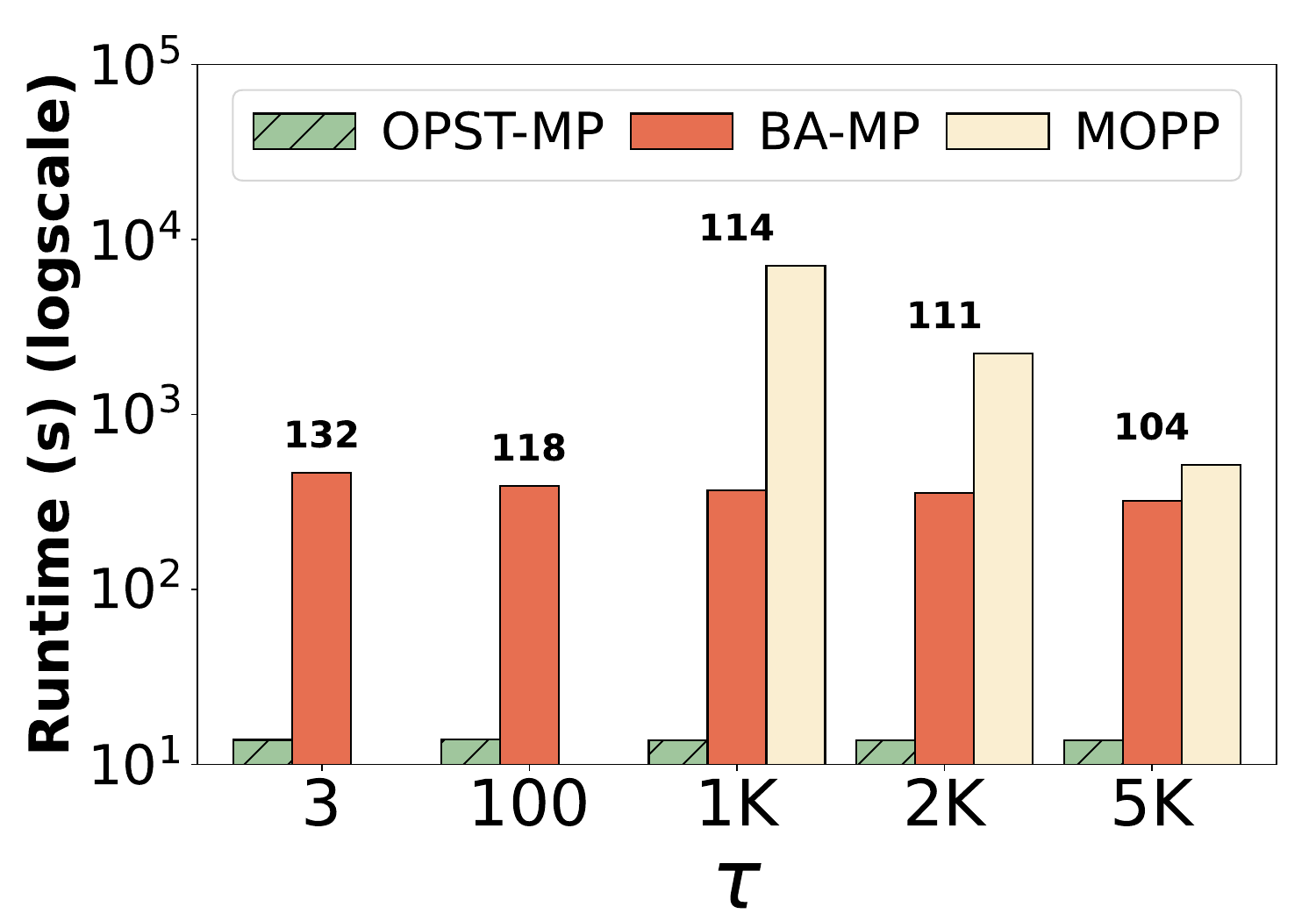}
    \caption{\SOL}\label{fig:runtime_mp_tau}
\end{subfigure}%
\hfill
\begin{subfigure}[t]{0.248\textwidth}
    \centering
    \includegraphics[width=1\linewidth]{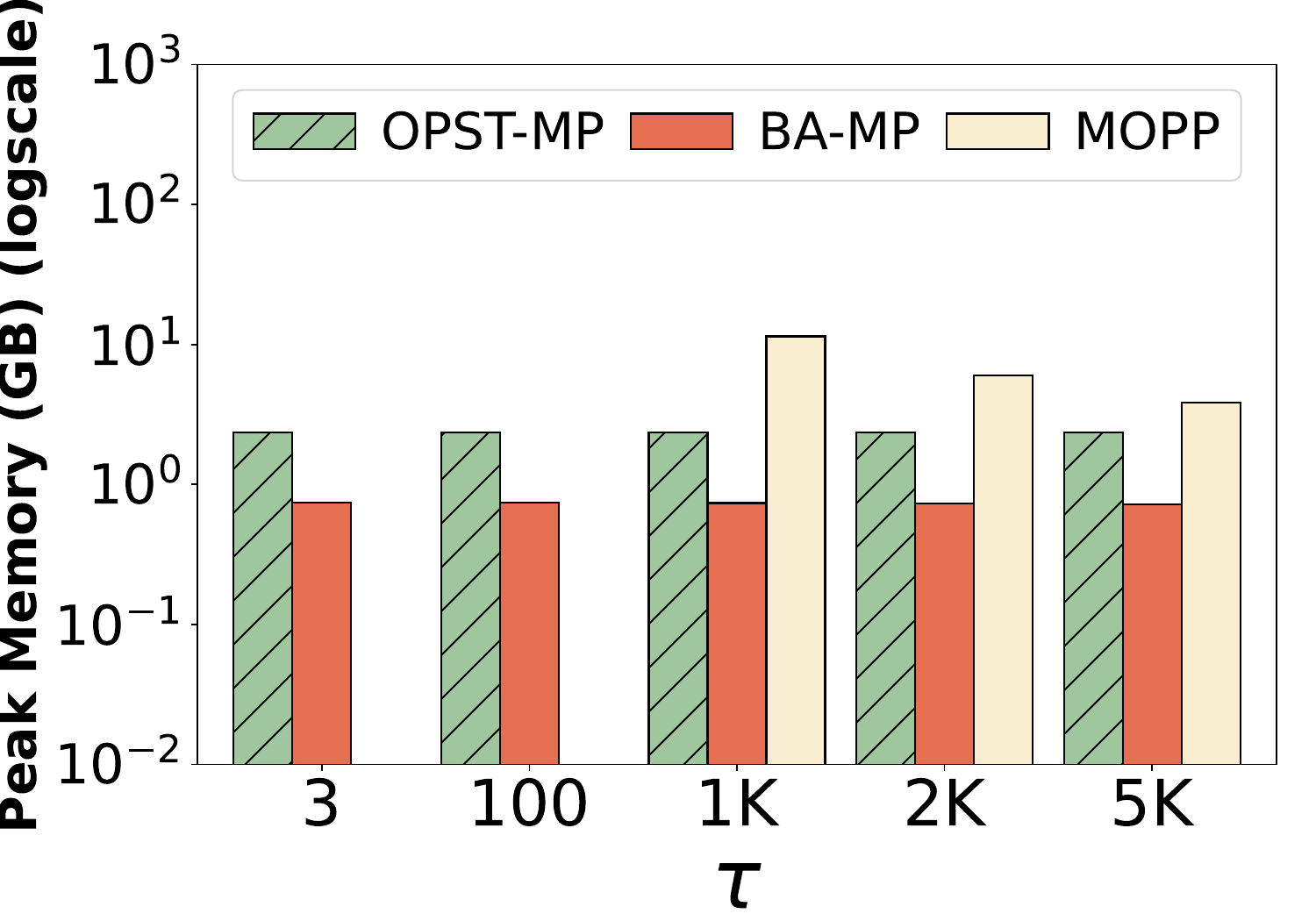}
    \caption{\SOL}\label{fig:runtime_mem_mp_tau}
\end{subfigure}%
\caption{\OPSTMP, \BAMP, and \MOPP: (a) Runtime and 
(b) peak memory consumption for varying $n$. (c) Runtime and 
(d) peak memory consumption for varying $\tau$. Missing bars for \BAMP and \MOPP indicate that they \emph{did not finish within 24 hours.} The value above each pair of bars in (a) and (c) represents the maximum length $k$ of all $\tau$-maximal $\tau$-frequent OP patterns.}

\centering 
\begin{subfigure}[t]{0.245\textwidth}
    \centering
    \includegraphics[width=\linewidth]{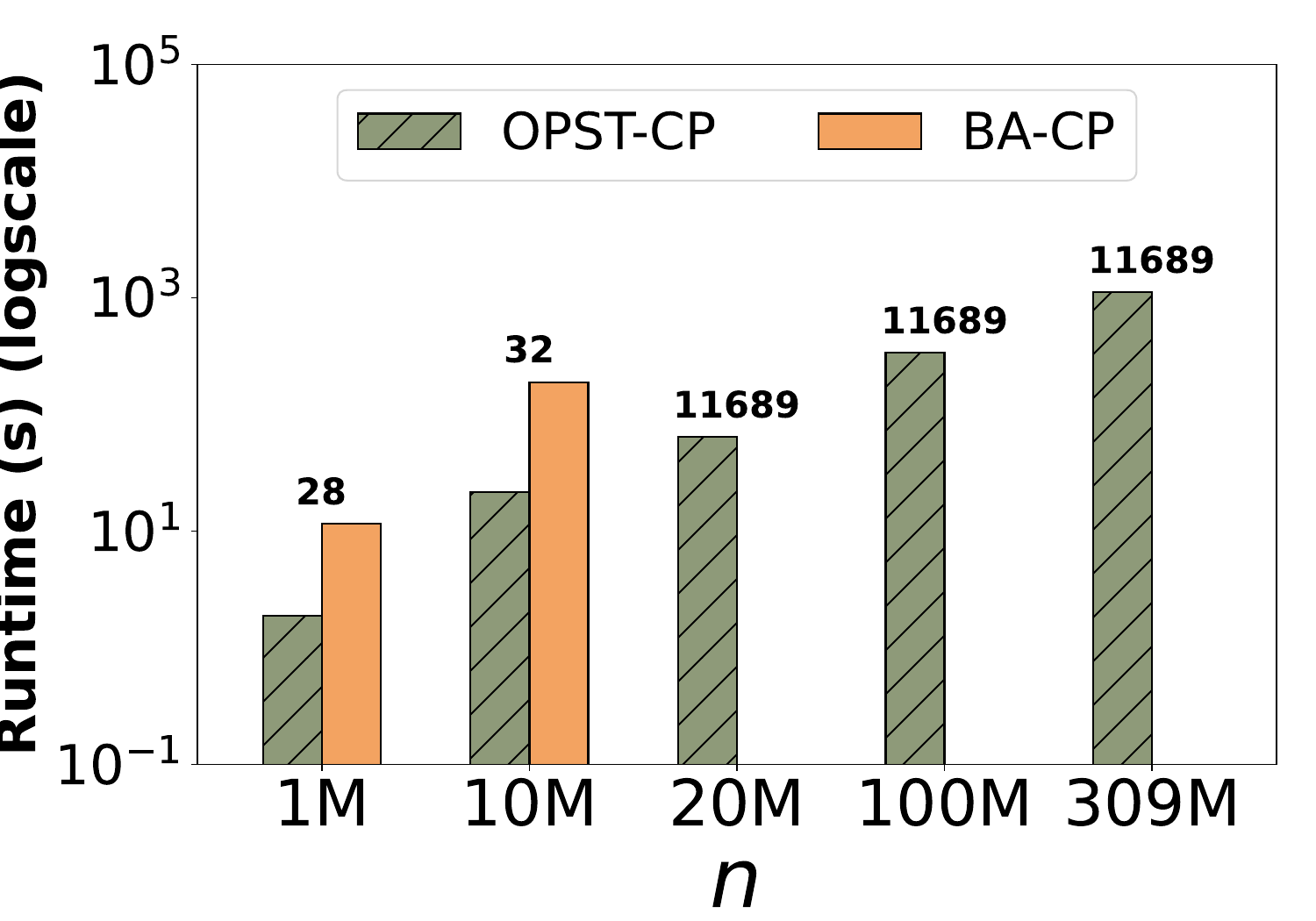}
    \caption{\WHA}\label{fig:opstcp:a}
\end{subfigure}
\hfill
\begin{subfigure}[t]{0.245\textwidth}
    \centering
    \includegraphics[width=\linewidth]{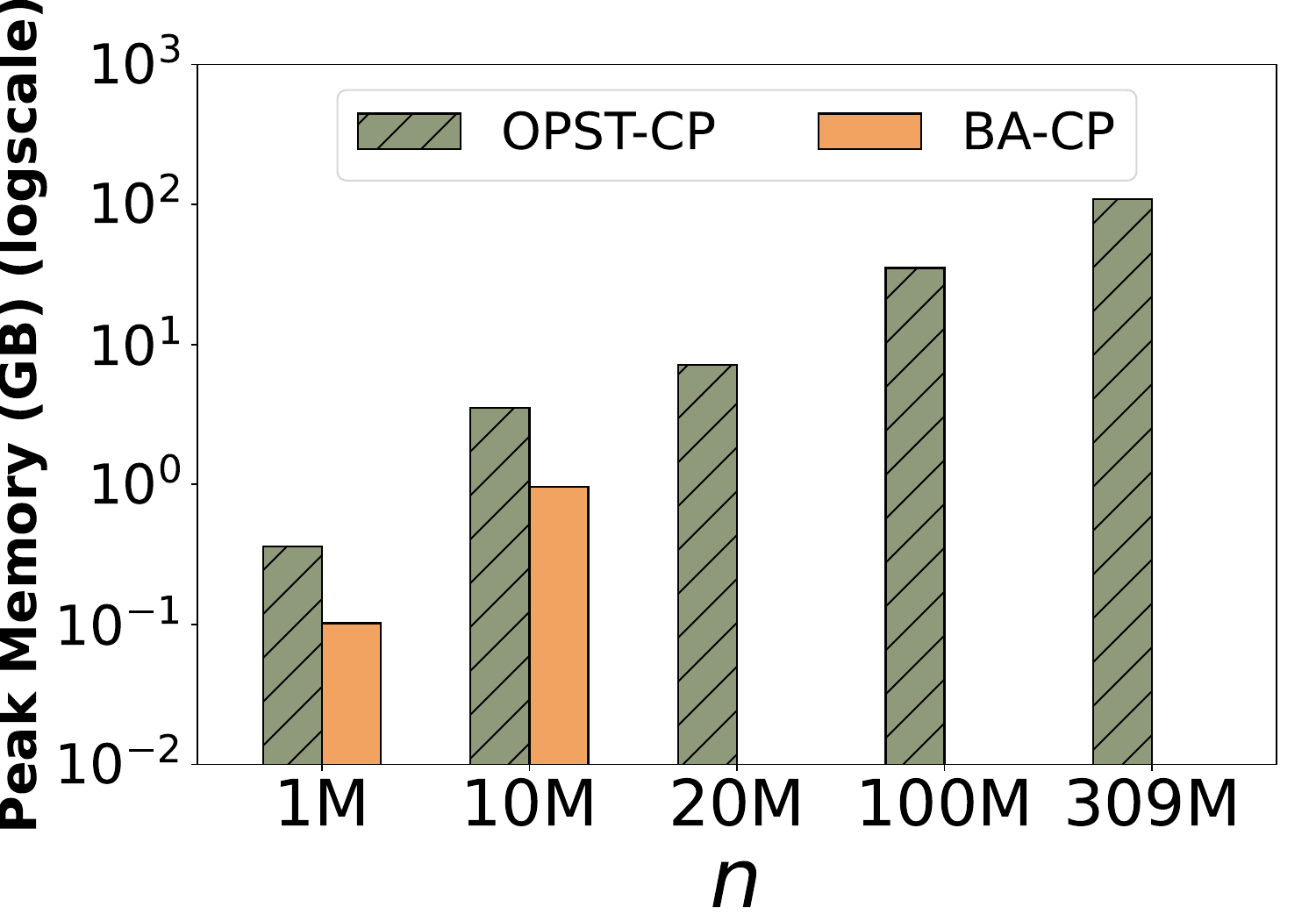}
    \caption{\WHA}\label{fig:opstcp:b}
\end{subfigure}
\hfill
\begin{subfigure}[t]{0.245\textwidth}
    \centering
    \includegraphics[width=\linewidth]{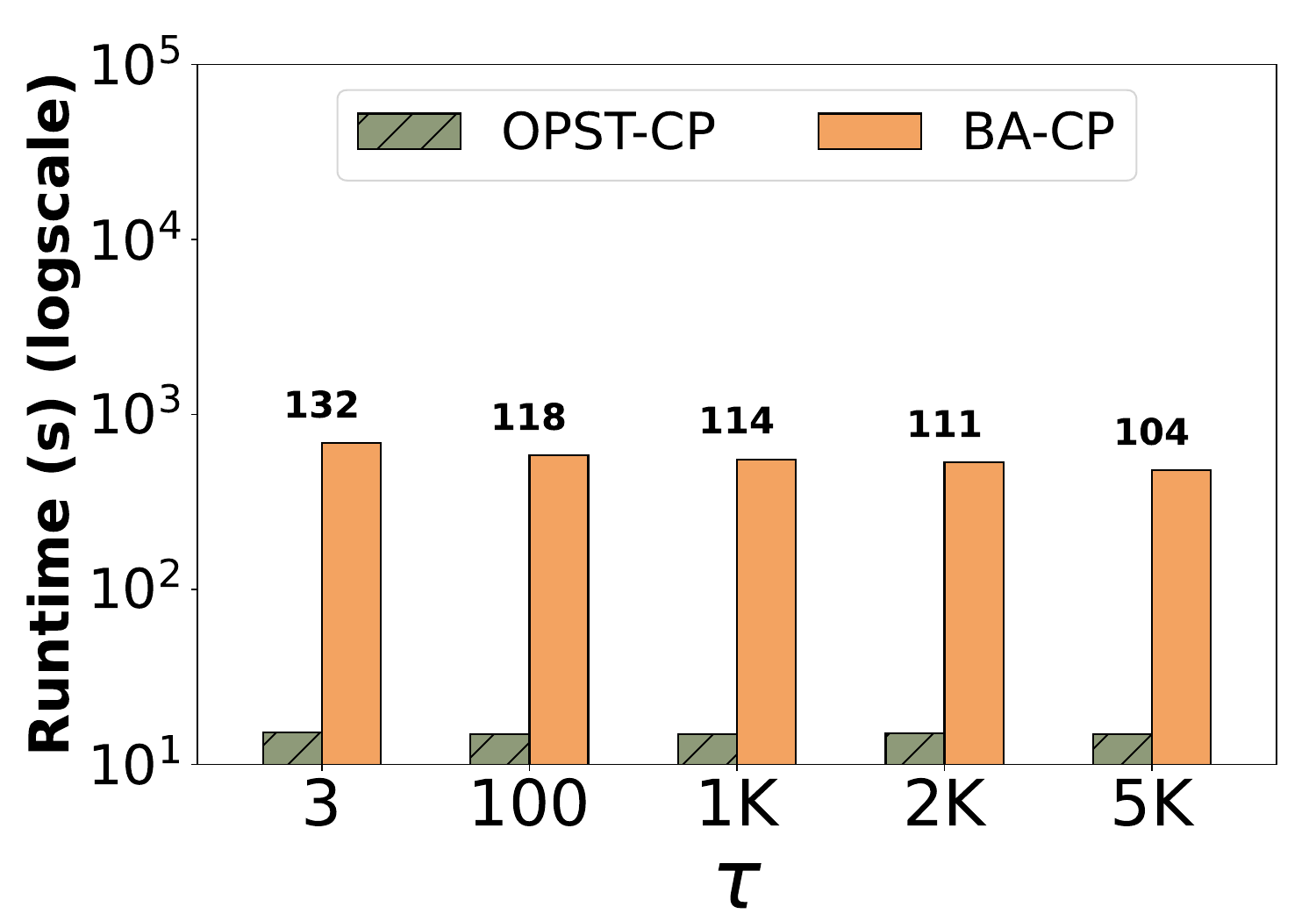}
    \caption{\SOL}\label{fig:opstcp:c}
\end{subfigure}%
\begin{subfigure}[t]{0.245\textwidth}
    \centering
    \includegraphics[width=\linewidth]{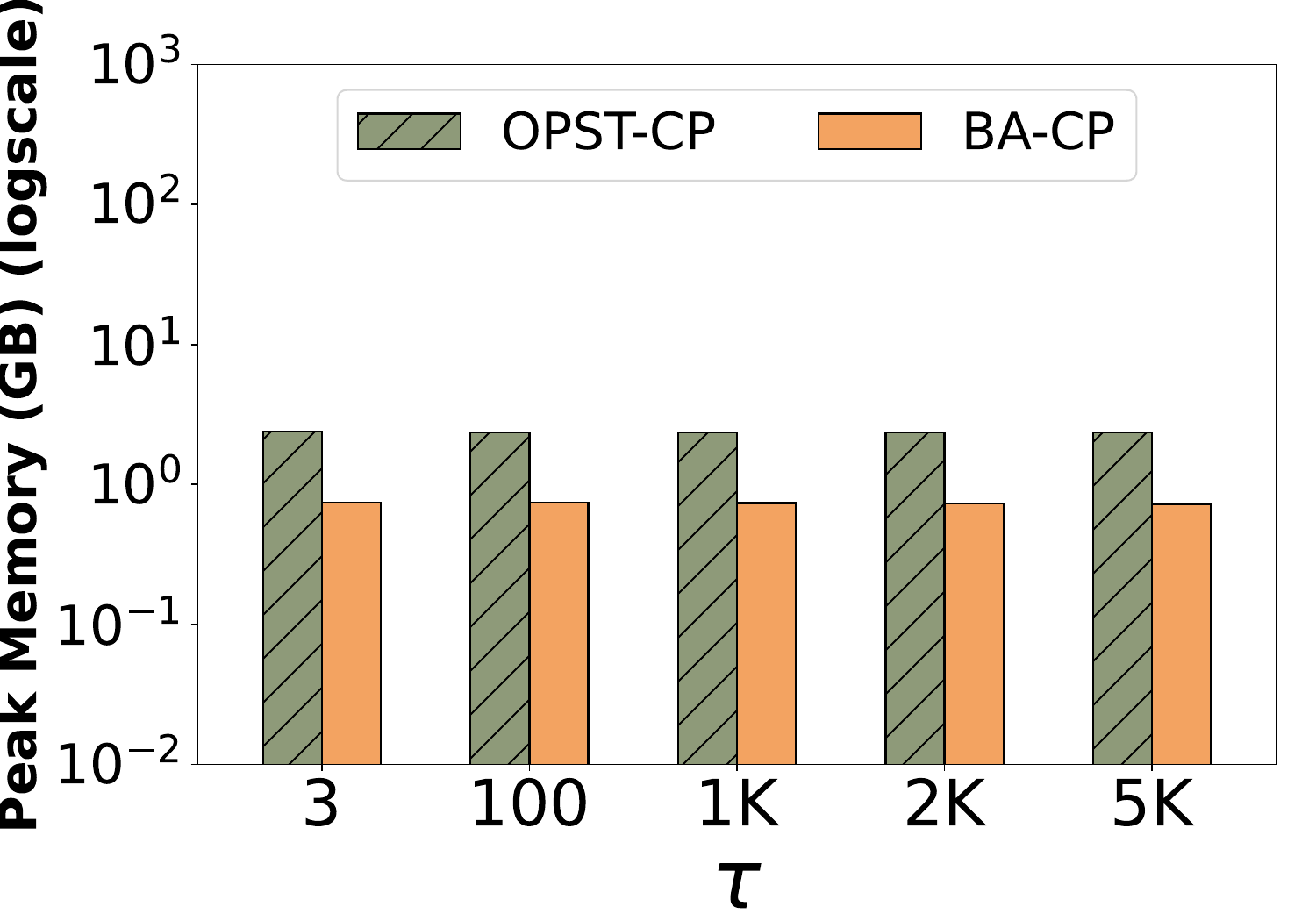}
    \caption{\SOL}\label{fig:opstcp:d}
\end{subfigure}%

 \caption{\OPSTCP and \BACP: (a) Runtime and 
(b) peak memory consumption for varying $n$. (c) Runtime and 
(d) peak memory consumption for varying $\tau$.  Missing bars for \BACP indicate that it \emph{did not finish within 24 hours.} The value above each pair of bars in (a) and (c) is  the maximum length $k$ of all closed $\tau$-frequent OP patterns.}\label{fig:opstcp}
\end{figure*}
\begin{figure}[!ht]
\centering 
\begin{subfigure}[t]{0.35\textwidth}
    \centering
    \includegraphics[width=1\linewidth]{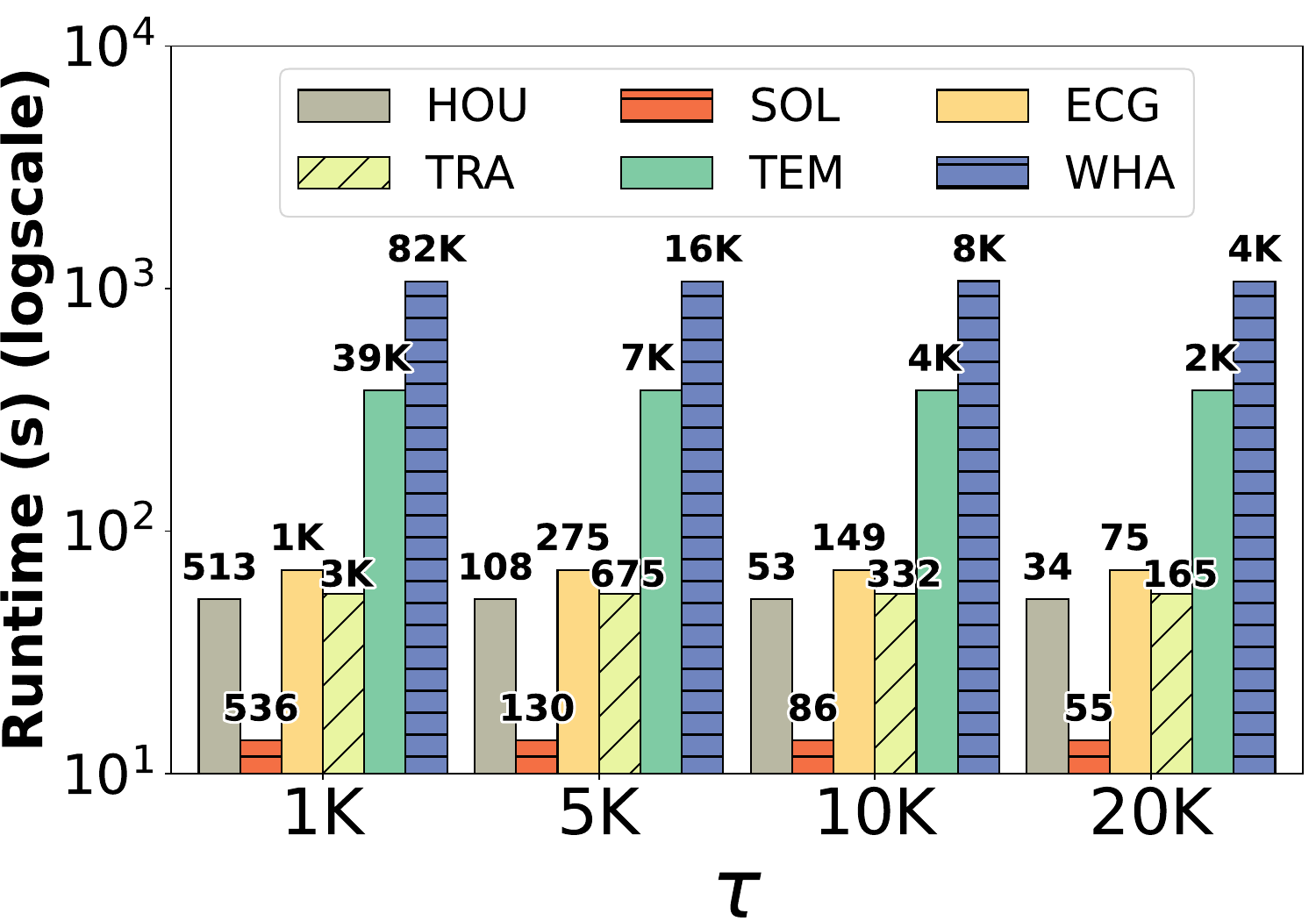}
    \caption{All datasets}\label{fig:opstmp:all_time}
\end{subfigure}%
\quad
\begin{subfigure}[t]{0.35\textwidth}
    \centering
    \includegraphics[width=1\linewidth]{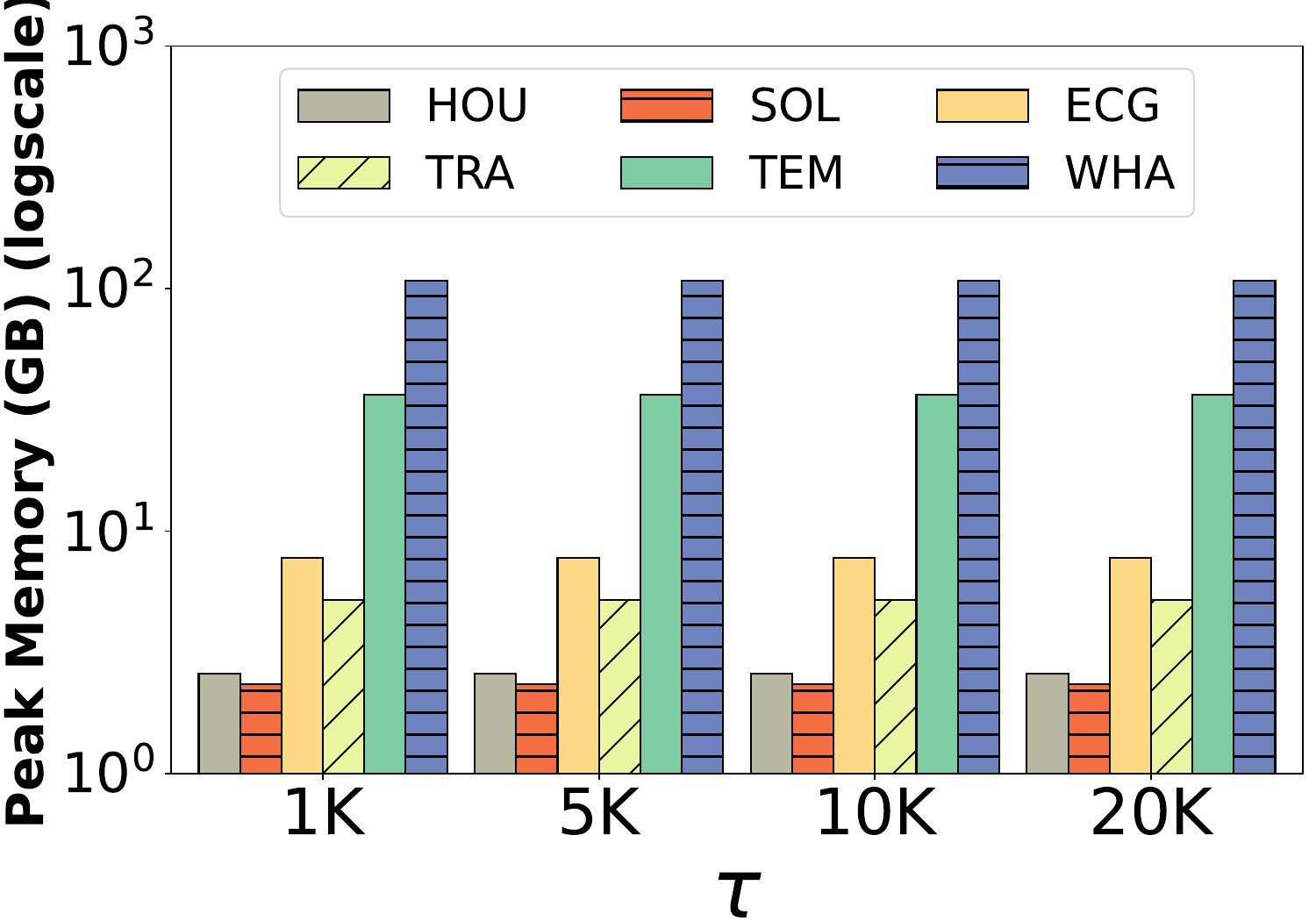}
    \caption{All datasets}\label{fig:opstmp:all_mem}
\end{subfigure}%
\caption{\OPSTMP: (a) Runtime and (b) peak memory consumption on all datasets for varying $\tau$.  The value above each bar in (a) is the number of $\tau$-maximal $\tau$-frequent OP patterns. \BAMP and \MOPP are omitted; they were slower than \OPSTMP by more than one order of magnitude on average.
}\label{fig:opstmp:all_tau}
\end{figure}

\paragraph{Maximal Frequent OP Pattern Mining.}~Figs.~\ref{fig:runtime_mp_n} and~\ref{fig:mem_mp_n} show the impact of $n$ on the runtime and space, respectively, of 
all algorithms. 
 \OPSTMP was \emph{orders of magnitude faster} than both  competitors and \emph{much more scalable with $n$}. For example, \OPSTMP required $1071$ seconds for the entire \WHA dataset ($n\approx 309\cdot 10^6$), whereas \BAMP and \MOPP did not finish within $24$ hours. Specifically, 
\BAMP did not finish within $24$ hours when $n\geq 20\cdot 10^6$ since $k > 11000$ ($k$ increased due to a long sequence of repetitions of the same letter), and \MOPP did not finish within $24$ hours when $n \geq 10\cdot10^6$. 
Interestingly, only $21.4\%$ on average (and up to $25.5\%$) of the runtime of \OPSTMP was spent to mine the patterns 
and the rest to construct the \OPST. This is encouraging as the \OPST can be constructed once and the mining can be performed several times with different $\tau$ values. 
In terms of space, both \OPSTMP and \BAMP scaled linearly with $n$, in line with their $\cO(n)$ space complexity, but the former required more space, due to the underlying index. However, \OPSTMP required \emph{more than one order of magnitude less space} than \MOPP, as expected (recall that the space complexity of the latter is $\Omega(n^2)$ in the worst case). 
Similar results were obtained for all other datasets, which have smaller $n$ and $\sigma$ than \WHA (see \cite{Supplement2024}). 

Figs.~\ref{fig:runtime_mp_tau} and \ref{fig:runtime_mem_mp_tau} show the impact of $\tau$ on the runtime and space, respectively, of all algorithms. As expected, the runtime of \OPSTMP was not affected by $\tau$, whereas 
both competitors needed more time for small $\tau$ values. 
Specifically, \BAMP was slightly slower for small $\tau$ because $k$ increased,  
and \MOPP was substantially slower (e.g., it did not finish within $24$ hours for $\tau\leq100$) due to a much larger $L$ 
(see its time complexity in Section~\ref{sec:related}).   
For example, $L=361$ when $\tau=5000$ and $L=1600$ when $\tau=1000$.   
Again, a small percentage of the runtime of \OPSTMP ($19.1\%$ on average) was  spent to mine the patterns. In these experiments,  
\OPSTMP required more space than \BAMP, due to its index, but less space than \MOPP, as expected by its space complexity. Similar results were obtained for all other datasets (see Fig.~\ref{fig:opstmp:all_tau}, which also shows the number of $\tau$-maximal $\tau$-frequent OP patterns). 

\begin{figure}[t]
\centering 
\begin{subfigure}[t]{0.35\textwidth}
    \centering
    \includegraphics[width=1\linewidth]{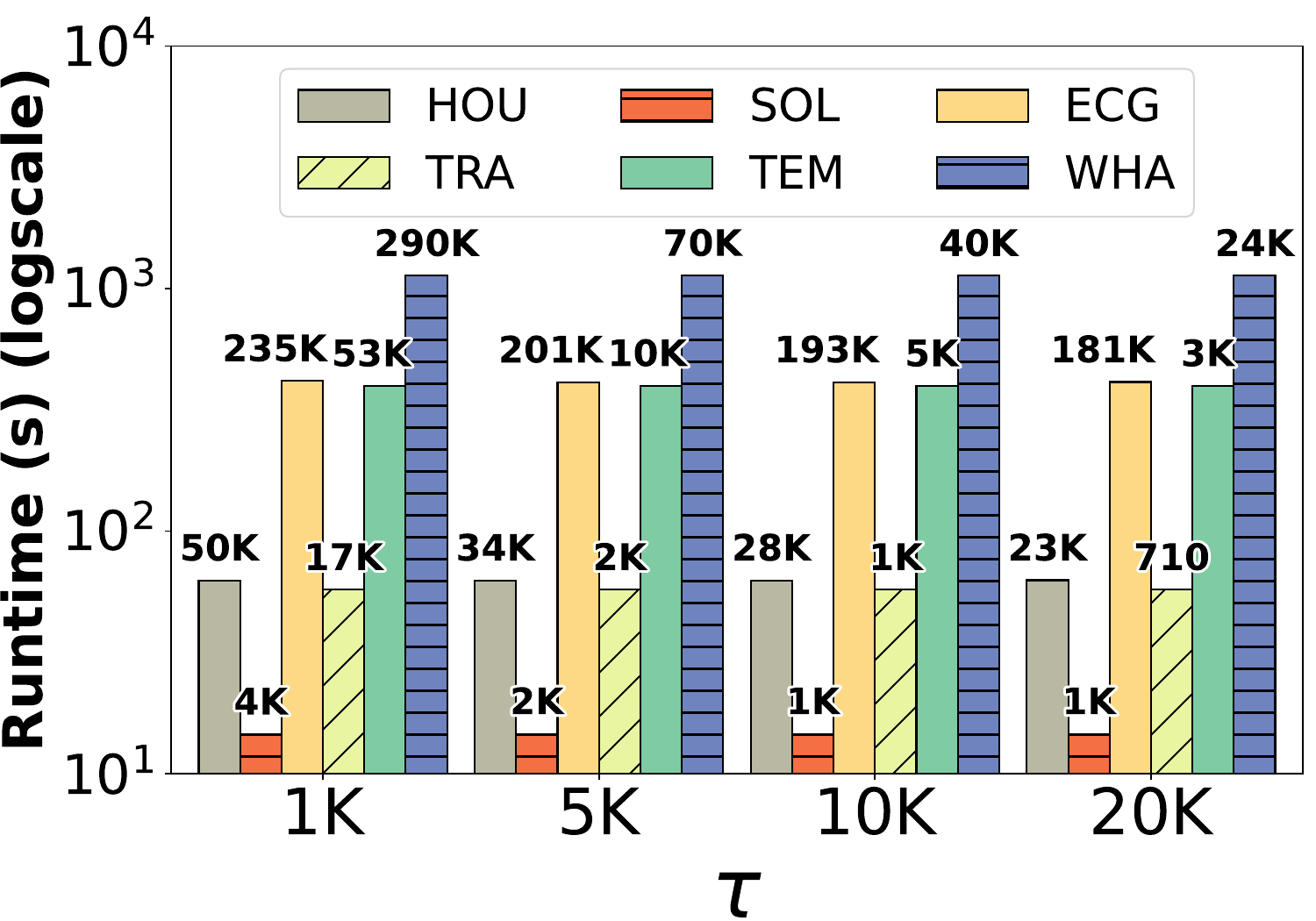}
    \caption{All datasets}\label{fig:opstcp:all_time}
\end{subfigure}%
\quad
\begin{subfigure}[t]{0.35\textwidth}
    \centering
    \includegraphics[width=1\linewidth]{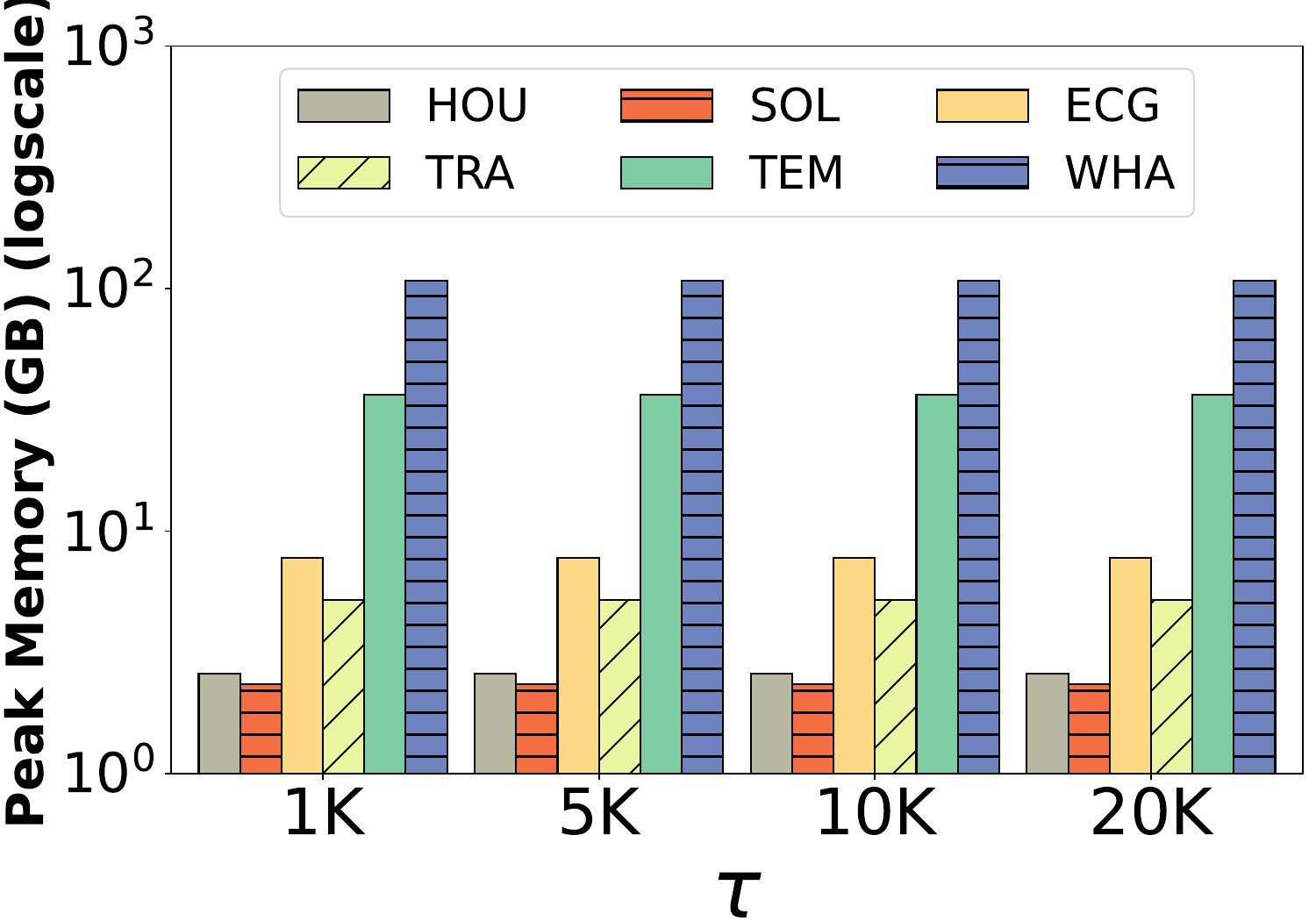}
    \caption{All datasets}\label{fig:opstcp:all_mem}
\end{subfigure}%
\caption{\OPSTCP: (a) Runtime and (b) peak memory consumption on all datasets for varying $\tau$.  The value above each bar in (a) is the number of closed $\tau$-frequent OP patterns. \BACP is omitted, as it was slower than \OPSTCP by more than one order of magnitude on average across all datasets.}\label{fig:opstcp:alldata}
\end{figure}

\paragraph{Closed Frequent OP Pattern Mining.} Figs.~\ref{fig:opstcp:a} and~\ref{fig:opstcp:b} show the impact of $n$ on runtime and space, respectively, of \OPSTCP and of the baseline \BACP. 
The results are in line with the time complexities of these algorithms. That is, \OPSTCP is much more scalable than \BACP. For example, when $n \approx 309\cdot 10^6$, \OPSTCP needed $1130$ seconds, while \BACP did not finish within $24$ hours when $n\geq 20\cdot 10^6$, as $k>11000$. Again, a relatively small percentage of the total time of \OPSTCP ($26\%$ on average and up to $31.5\%$) was spent for mining the patterns. In terms of space, both \OPSTCP and \BACP scaled linearly with $n$, in line with their space complexity, but \OPSTCP required more space, due to its index.  
Similar results were obtained for all other datasets 
(see \cite{Supplement2024}).

Figs.~\ref{fig:opstcp:c} and~\ref{fig:opstcp:d} show the impact of $\tau$ on runtime and space, respectively, of \OPSTCP and \BACP. The results are 
in line with the
time complexities of these algorithms. That is, the runtime of \OPSTCP was 
not affected by $\tau$ but that of \BACP was $1.4$ times larger for $\tau=3$ compared to $\tau=5000$ due to the larger $k$ ($132$ vs. $104$). 
Again, only $26\%$ of the total time of \OPSTCP on average was spent to mine the patterns and the rest for constructing the index. Furthermore, the space of \OPSTCP was larger than \BACP, due to its index. Similar results were obtained for all other datasets (see Fig.~\ref{fig:opstcp:alldata}). Despite being a superset of maximal frequent OP patterns, closed frequent OP patterns can still be mined very fast even for small $\tau$, thanks to our index. 
In Fig.~\ref{fig:opstcp:all_time}, the closed frequent OP patterns are $265$ times more on average compared to the maximal in Fig.~\ref{fig:opstmp:all_time} but their mining is less than $2$ times slower. 

\paragraph{Case Studies.}~We show the benefit of clustering based on maximal frequent OP patterns, following the methodology of~\cite{xindong_opp,tkde_opp}: we perform pattern-based clustering of a dataset comprised of multiple strings; the patterns are the features, the frequency of a pattern in a string is the feature value, and the $k$-means clustering algorithm is used to cluster the feature matrix. We confirm the findings of~\cite{xindong_opp,tkde_opp} on large-scale data: OP pattern-based clustering greatly outperforms clustering the raw data where each time series element is a feature. 

\begin{table}[ht]
\caption{Pattern-based vs raw data clustering.}\label{tab:clustering}

\centering
\begin{tabular}{|c||c|c|c|c|}
\hline
\textbf{Dataset}                & \textbf{Approach}       & \textbf{NMI}    & \textbf{h}      & \textbf{RI}     \\ \hline \hline
\multirow{2}{*}{\CCT}        & Raw Data Clustering          & 0.45  & 0.46   & 0.74   \\ \cline{2-5} 
                                     & Max. OP Pattern-based ($\tau=8$)  & \textbf{1}     & \textbf{1}      & \textbf{1}     \\ \hline
\multirow{2}{*}{\WAF}               & Raw Data Clustering           & $9\cdot10^{-7}$    & $1\cdot10^{-6}$      & 0.55   \\ \cline{2-5} 
                                     & Max. OP Pattern-based ($\tau=9$)  & \textbf{1}      & \textbf{1}     & \textbf{1}     \\ \hline
\end{tabular}
\end{table}

\begin{figure}[t]
\centering
\begin{subfigure}[t]{0.35\textwidth}
    \centering
    \includegraphics[width=\linewidth]{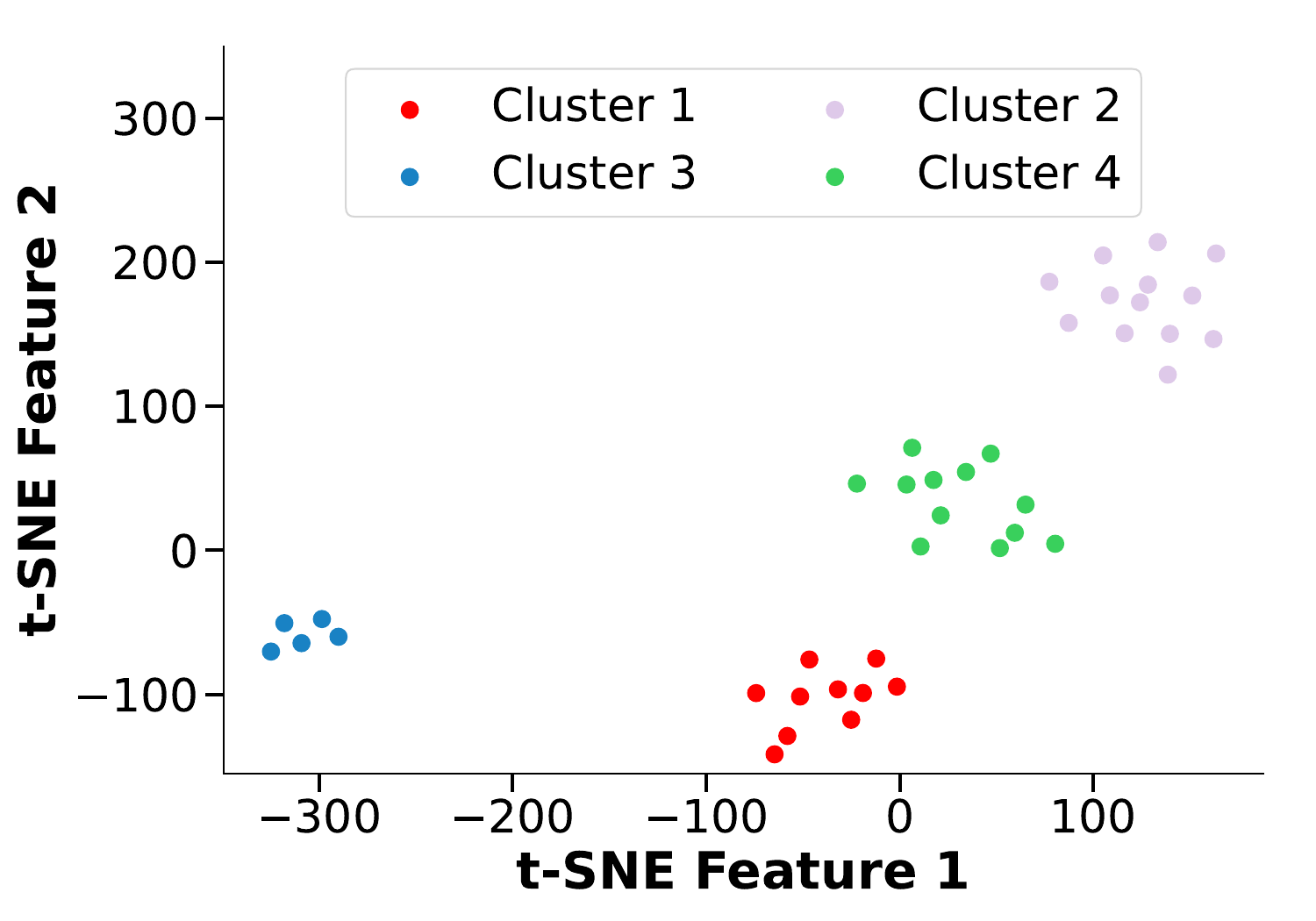}
    \captionsetup{skip=1pt}
    \caption{\CCT ($k=4$, $\tau=8$)}
\end{subfigure}%
\quad
\begin{subfigure}[t]{0.35\textwidth}
    \centering
    \includegraphics[width=\linewidth]{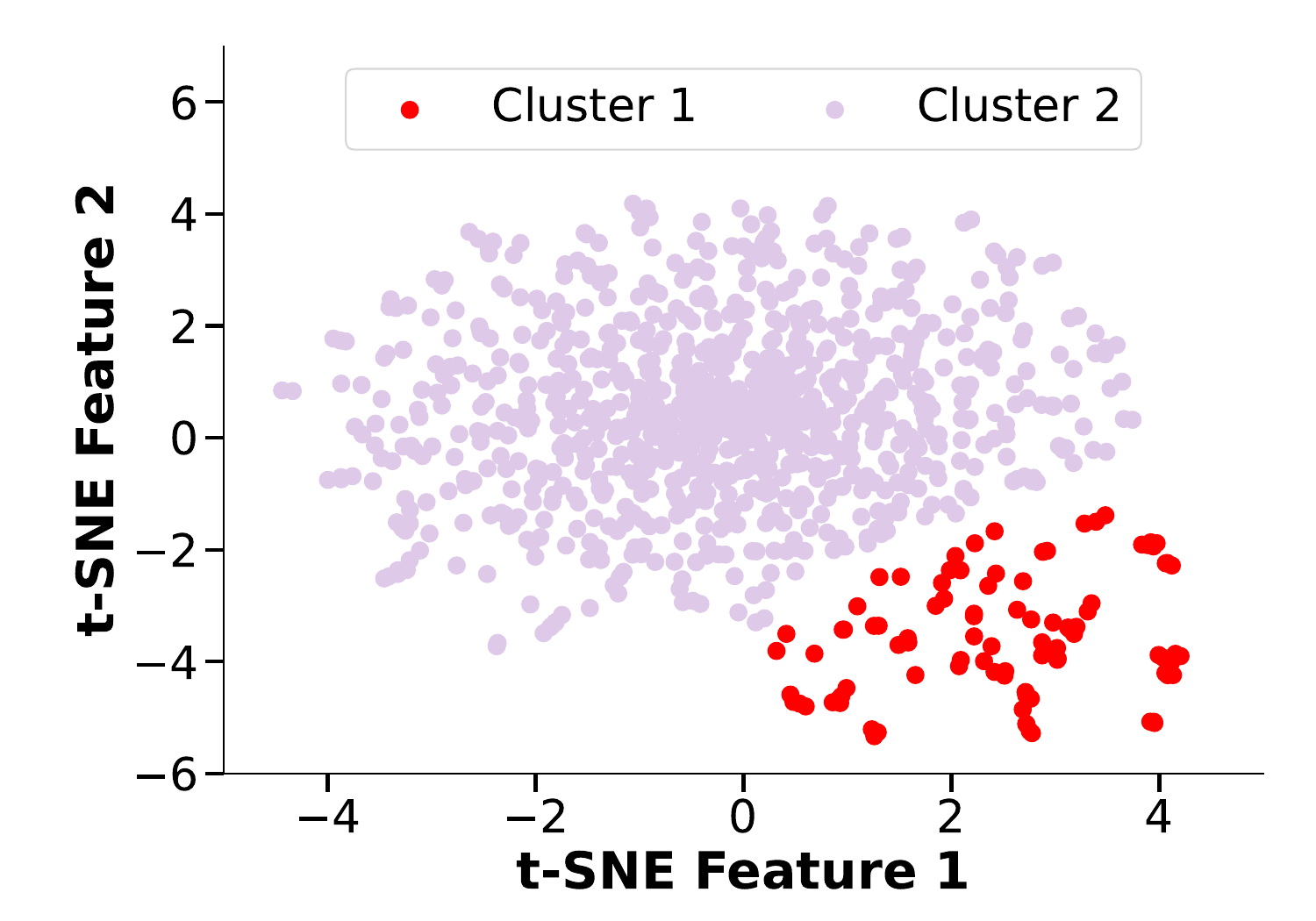}
    \captionsetup{skip=1pt}
    \caption{\WAF ($k=2$, $\tau=9$)}
\end{subfigure}%

\caption{t-SNE visualization of pattern-based clustering.} 
\label{fig:cluster_pub}
\end{figure}

First, we used $2$ UCR datasets, CinCECGTorso (\CCT) and Wafer (\WAF), which have ground truth. We evaluated how similar is the clustering obtained from pattern-based or  raw data clustering to the ground truth using three well-known measures:  Normalized Mutual Information (NMI), homogeneity (h) and Rand Index (RI) (see \cite{Supplement2024} for details). These measures take values in $[0,1]$ (higher values are preferred). 
Table~\ref{tab:clustering} shows that pattern-based clustering significantly outperformed the raw data clustering. Additionally, we used t-SNE~\cite{van2008visualizing} to project the pattern-based clustering results onto a $2$D space. The results in Fig.~\ref{fig:cluster_pub}  
illustrate 
the high quality of our approach. 

Then, we used our \ECG dataset, which was converted to a collection of strings, one per participant. We applied \OPSTMP 
with $\tau = 68000$ obtaining $25$ maximal frequent OP patterns (features), and $k$-means with $k=2$ to cluster the feature matrix. The clusters were evaluated by a domain expert and found to be coherent and well-separated, indicating a high-quality result, unlike those obtained by raw data clustering.

\begin{table}[!t]
\centering
\caption{Mean and/or standard deviation in 5 key attributes for the clusters obtained by pattern-based clustering.}\label{tab:maria}
\begin{tabular}{|l||c|c|}
\hline
 {\bf Attribute}                    & {\bf Cluster 1}   & {\bf Cluster 2}   \\ \hline \hline
A1. Pain Tolerance     & 4, 2        & 4.42, 2.09  \\ \hline
A2. Exhaustion Tolerance          & 11.13, 3.09 & 11.81, 2.66 \\ \hline
\begin{tabular}[c]{@{}l@{}} A3. Maximum Voluntary Contraction\end{tabular} & 12.06, 7.82 & 10.74, 5.02 \\ \hline
A4. Training Hours / Week & 3.13   & 1.96  \\ \hline
A5. Age             & 25.63, 8.12 & 22.06, 4.49 \\ \hline
\end{tabular}
\end{table}

As can be seen in Table~\ref{tab:maria}, the  participants in Cluster 1 are different from 
those in Cluster 2, along 5 key attributes, A1 to A5, and these differences were explained by the domain expert as follows. The differences in A1 and A2 may be because the participants in Cluster 1 had greater muscle strength (see A3) and were more active (see A4). Furthermore, the participants in Cluster 1 performed the task for 4 minutes and 35 seconds on average before withdrawing from the exercise entirely, while those in Cluster 2 lasted much longer and performed the task for 5 minutes and 45 seconds on average. Both clusters were gender-balanced. According to the domain expert, the results suggest that participants who exercise more are more attuned to bodily sensations, allowing them to protect their bodies from potentially damaging stimuli at an earlier stage. 

\section*{Acknowledgments}
This work is supported by the PANGAIA and ALPACA projects that have received funding from the European Union’s Horizon 2020 research and innovation programme under the Marie Skłodowska-Curie grant agreements No 872539 and 956229, respectively.
L.~L.~is supported by a CSC Scholarship.
We thank Jakub Radoszewski for pointing us to~\cite{DBLP:journals/siamcomp/ColeH03a}.

\bibliographystyle{plain}
\bibliography{references}

\end{document}